\documentclass[conference,compsoc]{IEEEtran}

\ifCLASSOPTIONcompsoc
  \usepackage[nocompress]{cite}
\else
  \usepackage{cite}
\fi

\usepackage{booktabs}
\usepackage{mathtools}
\usepackage{amssymb,amsmath,amsthm}
\usepackage{algorithm}
\usepackage{algpseudocode}
\algrenewcommand\algorithmicrequire{\textbf{Input:}}
\usepackage{setspace}
\usepackage{color}
\usepackage{hyperref}
\usepackage[normalem]{ulem}
\usepackage{soul}
\usepackage{xparse}
\setul{0.6ex}{0.1ex} % underline: drop a bit lower, slightly thicker

\newif\iffullversion
\fullversiontrue % Use \fullversionfalse for the submission version.
\NewDocumentEnvironment{fullproof}{+b}
{%
  \iffullversion
    \begin{proof}#1\end{proof}%
  \fi
}
{}

\NewDocumentEnvironment{showWhenSubmit}{+b}
{%
  \iffullversion
  \else
    #1%
  \fi
}
{}

\newtheorem{definition}{Definition}
\newtheorem{theorem}{Theorem}
\newtheorem{lemma}{Lemma}
\newtheorem{proposition}{Proposition}
\newtheorem{corollary}{Corollary}

\newcommand{\restatedhead}{}
\newtheorem*{restatedbox}{\restatedhead}
\newenvironment{restated}[3]
  {\renewcommand{\restatedhead}{#1~\ref{#2}}\begin{restatedbox}[#3, restated]}
  {\end{restatedbox}}

\newcommand{\proofloc}[1]{\iffullversion The proof is in Appendix~\ref{#1}.\fi}

\newcommand{\tx}{\textit{tx}}
\newcommand{\txset}{\textbf{TX}}
\newcommand{\psets}{\textbf{P}}

\newcommand{\alltx}{\txset_\textit{all}}
\newcommand{\winner}{f_{\textit{win}}}
\newcommand{\pubprotocol}{\mathcal{P}^\textit{pub}}
\newcommand{\bidvector}{\textbf{b}}
\newcommand{\txcontent}{\textit{data}}
\newcommand{\dtr}{\textit{dtr}}

\begin{document}
\title{Time Is Money: Incentivized Causal Transaction Ordering}
\sloppy

\author{\IEEEauthorblockN{Hongyin Chen}
\IEEEauthorblockA{Technion\\
hongyin.chen.contact@gmail.com}
\and
\IEEEauthorblockN{Xu Zheng}
\IEEEauthorblockA{Miami Herbert Business School\\
zhengxu@miami.edu}
\and
\IEEEauthorblockN{Jichen Li}
\IEEEauthorblockA{Tsinghua University\\
jichenli@mail.tsinghua.edu.cn}
\and
\IEEEauthorblockN{Ittay Eyal}
\IEEEauthorblockA{Technion\\
ittay@technion.ac.il}
}


% compsoc conferences compared to traditional conferences), use this
%Telephone: (800) 555--1212, Fax: (888) 555--1212}


\maketitle

\begin{abstract}
Front-running is a subtle and persistent problem for blockchains.
A blockchain is a stateful virtual machine executing instructions called transactions. 
Users earn rewards by publishing \emph{functional transactions} essential to the system.
Attackers observe these transactions and publish their own ahead of the users', seizing the reward and eroding users' incentive to publish functional transactions.
Preventing front-running means enforcing \emph{causality}: 
If an attacker receives transaction $\textit{tx}_A$ and then publishes transaction $\textit{tx}_B$, then $\textit{tx}_A$ must be ordered before $\textit{tx}_B$.
However, this causality is only observed by the attacker. 
Practical systems order transactions by bid amount, so transactions willing to pay more get executed first, but this only results in a bidding war eroding users' rewards.
Though numerous ordering approaches have been proposed, none achieves causality, leaving users vulnerable to front-running.

We present \emph{PRECEDE}, a mechanism-design approach that enforces transaction causality by removing the economic incentive to front-run.
PRECEDE orders transactions by a power-weighted randomized lottery, whose winning probability grows super-linearly in the bid.
The user's strategy of publishing a transaction with a deterring bid forms an equilibrium where the attacker refrains from competing.
Moreover, PRECEDE prevents the prominent sandwich attack, which relies on front-running. 
PRECEDE can be directly deployed in any censorship-resistant blockchain with a simple change to its transaction ordering mechanism.

% Front-running is a central problem in blockchain systems. 
% Those are stateful virtual machines executing an ordered list of instructions called transactions.
% Their users often issue functional transactions that benefit the system and earn rewards.
% But attackers front-run them: 
% Observing a transaction, they manipulate the order to get ahead, seizing the reward.
% This strips users of the incentive to issue these functional transactions, eroding system efficiency.
% This problem is hard. 
% Front-running is a matter of \emph{causality}, whether a sender issued her transaction after observing another's; but this causal relationship is private to the attacker and invisible to the blockchain. 
% Prior work suggested, among other approaches, ordering by bid amount, i.e., transactions willing to pay more get executed first, but this results in a race to the bottom. 

% We present \emph{PRECEDE}, a mechanism-design approach that, for the first time, enforces transaction causality by removing the economic incentive to front-run.
% PRECEDE orders transactions by a power-weighted randomized lottery whose winning probability grows super-linearly in the bid.
% A strategy where the user submits a transaction with a deterring bid forms an equilibrium where the attacker refrains from competing. 
% Moreover, PRECEDE prevents the prominent sandwich attack, which relies on front-running. 
% PRECEDE can be directly deployed on existing blockchains with a simple change to their transaction ordering mechanism. 
\end{abstract}


% \fi
\IEEEpeerreviewmaketitle

% \tableofcontents

\section{Introduction}\label{sec:intro}

  \emph{Blockchains} underpin the \$3 trillion cryptocurrency market~\cite{theblock2021marketcap3trillion} and the \$170 billion Decentralized Finance (DeFi) market~\cite{coindesk2025defitvl}.
  A blockchain is a stateful virtual machine that executes an ordered sequence of instructions called \emph{transactions}, each of which updates the machine's state.
  Blockchain validators~\cite{buterin2013ethereum,nakamoto2008bitcoin,yakovenko2018solana} extend this sequence by appending vectors of transactions called \emph{blocks}.
  Users issue many \emph{functional} transactions that benefit the system, such as arbitrage that aligns prices across markets~\cite{angeris2020improved}, among others~\cite{chen2025prrr,mouallem2026resilient}.
  They receive rewards from the system for these transactions.

  These rewards, however, are exposed to \emph{front-running}~\cite{daian2020flash}.
  When a user publishes a functional transaction, it is publicly visible before a validator places it on chain.
  An attacker observes it and can bribe the validator to order her own transaction ahead of the user's, seizing the reward~\cite{daian2020flash,torres2021frontrunner}.
  Front-running also serves as a building block for other order-manipulation attacks~\cite{zhou2021high,bartoletti2025theoretical,heimbach2022sok}.
  Practical chains often order transactions by bid~\cite{buterin2013ethereum,blackshear2024sui,yakovenko2018solana}, turning front-running into a race to the bottom: 
  % The user needs to bid the entire reward away to the validator~\cite{daian2020flash,mclaughlin2023large}.
  On Flashbots~\cite{li2023demystifying}, for instance, arbitrage winners pay validators a median of over $90\%$ of their reward~\cite{mclaughlin2023large}.
  Front-running is therefore harmful to blockchain functionality:
  With their rewards competed away, users have little incentive to issue functional transactions, and the markets that rely on them lose efficiency~\cite{angeris2020improved}.

  An extensive body of work (\S\ref{sec:related}) addresses these attacks by limiting the validators' ability to manipulate transaction order.
  One approach is to order transactions by their arrival times at an ordering committee~\cite{kelkar2020order,kelkar2023themis,kursawe2020wendy,zhang2020byzantine}, but this fails to prevent front-running~\cite{park2025frontrunning}:
  No such protocol ensures ordering an earlier transaction ahead of a later one~\cite{kelkar2020order,kursawe2020wendy}.
  A second approach makes the order within each block random~\cite{xrplf2022canonicaltxset,kavousi2026blindperm}.
  However, an attacker can publish many copies so that one likely precedes the target~\cite{mazorra2026timing,geth21350_random_ordering_spam_2023}; such spam burdens the system while still allowing front-running~\cite{mazorra2026timing,wang2026blockspace,tumas2023ripple}.
  A third approach encrypts transaction content to hide it from attackers~\cite{zhang2022f3b,stathakopoulou2021adding,chiang2022fairpos,canidio2024commitment,bormet2025beat}.
  However, it still fails to prevent front-running~\cite{garimidi2025limits}.

  Addressing front-running is fundamentally hard.
  It is a matter of \emph{causality}~\cite{lamport1978time}:
  If a participant creates and publishes her transaction after observing another's~\cite{reiter1994securely,cachin2001secure,duan2017secure}, it should appear later on chain.
  But this causal relationship is private to the attacker who issues the later transaction.
  Moreover, in a decentralized system we cannot compel participants to follow the protocol, for instance to refrain from front-running or from misreporting transaction arrival times.
  Instead, each participant is free to take whichever action maximizes her own revenue~\cite{tsabary2021mad,brugger2023checkmate,yaish2023uncle,mirkin2020bdos,carlsten2016instability,eyal2018majority,chen2025prrr}.
  Without making trust assumptions that nodes may violate for gain, the protocol can therefore rely only on a limited set of signals observable to the blockchain.
  Even the size of the front-running reward is unknown to the protocol and varies across orders of magnitude~\cite{mclaughlin2023large}, so any defense must hold across the full range.
  Together, these constraints force the protocol to defend blindly: against any rational adversary, at any revenue scale, using only signals observable to the blockchain.
  
  % The most effective response is to design an incentive structure under which no participant, whether user or attacker, profits from manipulating transaction ordering~\cite{tsabary2021mad,mirkin2020bdos,chen2025prrr}.

  In this work, we propose an ordering rule that prevents front-running under these constraints.
  We model~(\S\ref{sec:model}) front-running as a competition between rational participants, a \emph{regular user} and a set of \emph{attackers}, all maximizing their own expected revenue.
  The state of the system gives rise to an \emph{opportunity}.
  Any participant can publish a transaction to take advantage of this opportunity, and the first such transaction on chain earns a reward~$R$.
  Each transaction carries a bid; the winning transaction pays its full bid, while every later transaction pays $\gamma b$, where $\gamma > 0$ is the \emph{losing-fee rate}.
  The regular user discovers an opportunity first and publishes a transaction to earn the reward.
  An attacker observes it and reacts with a competing transaction.
  The two then keep publishing transactions or updating the bids of their existing transactions.
  An \emph{ordering protocol} then orders these transactions, determining the winner and the revenues.

  Our primary goal is to enforce \emph{causal ordering} of transactions.
  Most transactions are independent and therefore commute.
  We focus on conflicting ones, whose order affects outcomes; in our setting, these are the transactions that compete for the same opportunity.
  The protocol should also be \emph{profitable}: The user earns a strictly positive expected revenue since she finds the opportunity first; otherwise the user would have little incentive to issue functional transactions like arbitrage.
  Since a participant could publish many copies of her transaction to increase her chance of winning, we require \emph{anti-spam}: Each rational participant publishes at most one transaction.

  We present \emph{PRECEDE}, \ul{P}ower-weighted \ul{R}an\-dom\-i\-za\-tion for \ul{E}n\-forc\-ing \ul{C}au\-sal ordering through \ul{E}n\-try \ul{DE}ter\-rence~(\S\ref{sec:precede}).
  PRECEDE removes the incentive that drives front-running. 
  The attack occurs only because it is profitable; PRECEDE applies mechanism design to make it unprofitable, deterring attacks.
  It thus enforces the causal order without trust assumptions.

  To achieve this, PRECEDE orders transactions by a randomized power-weighted lottery:
  It gives each transaction with bid $b$ the weight $w(b) = b^k$, where $k > 1$, and places it first with probability proportional to that weight.
  This rule generalizes two existing mechanisms.
  As $k \to \infty$, all weight concentrates on the highest bid, degenerating to the bid-priority ordering of practical blockchains; at $k = 0$, all weights are equal, degenerating to uniform random ordering.

  Under this mechanism, the user has a \emph{deterrence strategy}: By publishing one transaction with a sufficiently large bid, she ensures that every later entrant who observes her transaction earns a non-positive expected revenue.
  Intuitively, the super-linear weight squeezes counter-bids from both ends: A small bid has too little weight to win, while a large bid costs more than its expected reward.
  Front-running is therefore unprofitable.
  We give a closed-form expression for the smallest deterring bid for a newly discovered opportunity.

  Order randomization raises the same spam concern as uniform random ordering: A participant might publish many transactions to increase her chance of winning.
  We show that every rational participant publishes a single bid if and only if $k \geq \ln 2 / \ln(1+\gamma)$, so PRECEDE satisfies anti-spam.
  We further show that, when $k$ meets the same bound, the deterrence strategy is also \emph{profitable}: The user earns strictly positive expected revenue and remains incentivized to discover opportunities.

  Beyond the existence of this strategy, we also verify that rational players are incentivized to play it rather than use other strategies.
  We analyze this as a sequential game (\S\ref{sec:game}), played between a user and an attacker.
  They act alternately, and the attacker moves last, reflecting her advantage in observing and reacting to others' transactions (e.g., from lower network latency, or from controlling the set of placed transactions as the validator).
  We prove that, whenever $k \ge \max\{2,\exp(1/\gamma)/\gamma\}$, it is a Subgame-Perfect Nash Equilibrium (SPNE) for the user to play the deterrence strategy in the first step and for neither player to act thereafter, achieving causal ordering regardless of the game's length.
  This bound is loose: Our numerical experiments show that the anti-spam bound alone already suffices.

  We further study how to choose $k$ so that PRECEDE meets both the anti-spam and equilibrium requirements while maximizing the user's expected revenue (\S\ref{sec:tradeoff}).
  A numerical simulation shows that the equilibrium bound is loose, so a much smaller $k$ already suffices in practice.
  Although $\gamma$ varies across opportunities, the designer can conservatively fix a single $k$ that works for a range of $\gamma$ values.
  For a fixed $\gamma$, the user's revenue is maximized at $k = 1 + 1/\gamma$, where it equals $\gamma R/(1+\gamma)$.

  Finally, we show that PRECEDE also defends against the \emph{sandwich attack}~\cite{zhou2021high}, a prominent front-running-based attack (\S\ref{sec:sandwich}).
  In a sandwich attack, a user buys an asset on a decentralized exchange, where buying raises the asset's price.
  The attacker front-runs to buy first, lets the user's own transaction push the price higher, and then back-runs to sell at that higher price, profiting at the user's expense.
  Under PRECEDE, defending against the sandwich attack is no harder than defending against front-running, and the user's resulting cost of deterrence is strictly less than the loss she would suffer from a sandwich attack in practical blockchains.

  In all, our contributions are:
  (1)~Modeling of front-running as competition under generic ordering protocols;
  (2)~entry deterrence as a new ordering principle, distinct from time- and content-based ones;
  (3)~PRECEDE, a power-weighted randomized lottery realizing this principle;
  (4)~a closed-form expression for the smallest bid that deters all front-running;
  (5)~an anti-spam reduction showing that single bids dominate multi-bid strategies;
  (6)~game-theoretic proof that the deterrence strategy is a subgame-perfect equilibrium;
  (7)~analysis of how to choose the parameter $k$ and
  (8)~extension to sandwich attacks.

  PRECEDE can be deployed directly in any censorship-resistant blockchain with only a small change to its transaction ordering mechanism.

  %%%%%%%%%%%%%%%%%%%%%%%%%%%%%%%%%%%%%%%%%%%%%%%%%%%%%%%%%%%%%%%%%%%%%%%%%%%%%%%%
  %%%%%%%%%%%%%%%%%%%%%%%%%%%%%%%%%%%%%%%%%%%%%%%%%%%%%%%%%%%%%%%%%%%%%%%%%%%%%%%%
  %%%%%%%%%%%%%%%%%%%%%%%%%%%%%%%%%%%%%%%%%%%%%%%%%%%%%%%%%%%%%%%%%%%%%%%%%%%%%%%%
  \section{Related Work}\label{sec:related}

  Front-running is a central security concern for blockchain systems.
  Such attacks are widespread in practice as attackers front-run user transactions for profit~\cite{daian2020flash,torres2021frontrunner,tumas2023ripple,mclaughlin2023large}.
  A large body of work studies front-running attacks~\cite{daian2020flash,byers2022combating,zecirovic2024analysis,tumas2023ripple,zhang2023front,zhang2025no,zhang2023time,zhang2024transfront,qin2023blockchain,wang2024frontrunning}, as well as other attacks they enable~\cite{heimbach2022sok}, such as sandwich attacks~\cite{zhou2021high,heimbach2022eliminating,xavier2023credible,bartoletti2025theoretical,li2023demystifying,torres2021frontrunner,li2024geth,adams2024don,ji2024regulatory} and multi-victim generalizations~\cite{wang2023n,bartoletti2022maximizing}.

  An extensive body of work designs transaction ordering protocols to defend against front-running~\cite{baum2022sok,raikwar2023fairness}.
  The dominant approach relies on a committee of nodes that orders transactions by their arrival order~\cite{kelkar2020order,kelkar2023themis,baird2016swirlds} or arrival timestamps~\cite{kursawe2020wendy,zhang2020byzantine}, aiming to prevent ordering manipulation.
  However, with an adversarial network and even a single corrupted committee node, no such protocol can respect the order in which nodes receive transactions~\cite{kelkar2020order,kursawe2020wendy}. 
  These local orders can conflict cyclically, leaving no global order that satisfies them all~\cite{gehrlein1983condorcet,vafadar2023condorcet}. 
  The resulting arbitrary tie-breaks let an attacker order herself ahead, and thus do not prevent front-running.
  The field has therefore relaxed its goal to weaker notions: \emph{batch-order} fairness, which groups transactions into batches without intra-batch order guarantees~\cite{kelkar2020order,kelkar2022order,kelkar2023themis,cachin2022quick,nagda2024rashnu,nagda2025dag,hay2025optimistic,wang2025dikaios,ren2025proof,kang2025fairdag,cachin2024quick} (or with stricter intra-batch rules that still admit front-running~\cite{yu2025efficient,ramseyer2024fair}); \emph{timed-order} fairness, which only constrains transactions separated by more than a certain delay threshold~\cite{kursawe2020wendy,kursawe2021wendy,zhang2020byzantine,Constantinescu2023fair,kang2025fairdag,zhang2024chronos,gramoli2024aoab,zarbafian2023aion,wang2022phalanx} (or provides this guarantee only with high probability~\cite{xue2024travelers}); and \emph{bounded unfairness}, which caps how far a transaction can be displaced behind another~\cite{kiayias2024ordering}.
  Yet none eliminates front-running: An attacker reacting quickly enough fits within the batch, time window, or position bound each relaxation permits. 
  The \emph{Ambush} attack empirically demonstrates this for a batch-order-fair system~\cite{park2025frontrunning}.

  Another line of work hides transaction content until ordering is fixed.
  These designs differ in who opens the ciphertext: a \emph{decrypting committee}~\cite{agarwal2025weighted,zhang2022f3b,yakira2021helix,malkhi2022maximal,bormet2025beat,bormet2025beast,cai2025equibft,zarbafian2023lyra,ciampi2025universally,momeni2022fairblock,choudhuri2025practical,choudhuri2024mempool,riva2025seahorse,gramoli2024aoab}, a \emph{trusted execution environment} (TEE~\cite{Sabt2015Trusted})~\cite{stathakopoulou2021adding,li2025transaction,ciampi2024universal}, a time-locked puzzle or delay-based encryption~\cite{chiang2022fairpos,rodriguez2025cryptographically,khajehpour2023mempool,agrawal2025timelock}, a trusted relay~\cite{babel2024prof}, or the transaction creator~\cite{canidio2024commitment,mcmenamin2022fairtradex,mcmenamin2024private}.
  However, such schemes still leave room for front-running: Publishing a transaction already exposes timing, propagation, and creator metadata, signals from which an attacker can speculatively race~\cite{garimidi2025limits}.
  Moreover, parties able to open the ciphertext may be incentivized to abuse that capability and, even when confidentiality holds, distort DeFi market microstructure such as liquidations~\cite{rondelet2023mempool}.
  Within the committee-decryption subclass, the failure runs deeper still: Wadhwa et al.~\cite{wadhwa2024data} prove a general impossibility with rational participants, where committee members can always deniably collude to leak contents.
  In contrast, we make transaction content public and reshape the ordering rule to disincentivize front-running.

  Some solutions hide transaction contents and offer weaker notions of order fairness at the same time~\cite{kursawe2021wendy,ciampi2025universally,zarbafian2023lyra,ciampi2024universal,gramoli2024aoab}.
  But they are as vulnerable to front-running as the receive-order-fair protocols and content-hiding schemes they combine.

  Within a single block, it is possible to uniformly randomize the transaction order to remove the deterministic priority that reordering attacks exploit~\cite{xrplf2022canonicaltxset,kavousi2026blindperm}.
  However, uniform randomization incentivizes spam, where adversaries publish many transactions to increase their probability of front-running~\cite{mazorra2026timing,wang2026blockspace,tumas2023ripple} or sandwich attacks~\cite{tumas2024ammazing}, as observed in deployed systems~\cite{geth21350_random_ordering_spam_2023}.
  PRECEDE disincentivizes spamming.

  Several works address sandwich attacks specifically in Automated Market Makers~\cite{chitra2022differential,canidio2023batching,chan2024mechanism,xavier2023credible,wadhwa2024data,alpos2023eating,Zhou2021A2MMMF,heimbach2022eliminating,jiang2025armm} and address specific, predefined front-running attacks that do not change in response to their defense~\cite{zhang2023front,park2025frontrunning,zhao2025mitigating}.
  A complementary direction lets the applications themselves impose sequencing constraints, rather than fixing a global ordering rule~\cite{durvasula2026monotone}.
  We address the general front-running problem.

  Several adjacent lines of work study transaction ordering without directly preventing front-running.
  Some target other goals: incentivized inclusion of \emph{report} transactions~\cite{chen2025prrr}, bounded confirmation delay~\cite{ahuja2023myochain}, reduced tail latency for long-waiting transactions~\cite{sokolik2020age}, empirical quantification of order-fairness violations under adversarial reordering~\cite{mahe2024adversary,mahe2025order}, similar order positions for transactions with similar issue times and bids~\cite{cohen2025fair,zhang2025ordered}, or coordination of transaction order with participants' declared workflow intent in collaborative processes~\cite{atwi2024transparent}.
  Another line detects unfair orderings post hoc rather than preventing them~\cite{nasrulin2023accountable}.
  But such auditing mis-accuses honest miners with probability exceeding 25\% within the sub-30-second window in which front-running operates~\cite{albrecht2026effectiveness}.

  Our mechanism enforces \emph{causal ordering}, a classical notion from distributed systems~\cite{reiter1994securely,cachin2001secure,duan2017secure}: 
  A transaction issued in response to an observed one should not be ordered before it.
  However, classical solutions (like Lamport timestamps~\cite{lamport1978time}) do not apply because malicious participants can misreport their timestamps to front-run, and the blockchain cannot verify them.
  We focus on the ordering of transactions for the same opportunity, which are linked by a data dependency so that their relative order affects their execution outcomes.
  Some ordering protocols~\cite{nagda2024rashnu,nagda2025dag} and practical systems~\cite{blackshear2024sui} also focus on such transactions, but require the blockchain to explicitly extract data dependencies between transactions.
  Extracting these dependencies is hard in general~\cite{gelashvili2023block}. 
  PRECEDE does not rely on transaction semantics.

  Another line of work, like ours, orders transactions by their bids.
  Validators in deployed blockchains usually order transactions by descending bid~\cite{buterin2013ethereum,blackshear2024sui,yakovenko2018solana}. 
  This introduces a bidding war: The user must outbid the attacker to win, paying away her reward~\cite{daian2020flash,mclaughlin2023large}.
  Flashbots~\cite{li2023demystifying} moves this competition off-chain, but the same bidding war still occurs~\cite{mclaughlin2023large}.
  We instead order at random.
  Another proposal lets users pay to improve their transactions' positions through a trusted party~\cite{mamageishvili23buying}; we do not rely on this strong assumption.

  In the game theory literature, contest theory~\cite{tullock1980efficient,skaperdas1996csf} studies similar competitive settings: Players pay bids to compete for a prize, and the winner is drawn at random, with each bid's winning probability proportional to a weight that increases with the bid. 
  Our mechanism adopts the power-weight function from the Tullock contest~\cite{skaperdas1996csf}. 
  Tullock contests, however, are \emph{full-pay}: Every player pays her bid in full whether she wins or loses. 
  In practical blockchains, by contrast, a losing transaction is reverted and consumes different computational resources than the winner, paying a ratio of its bid. 
  We adopt this partial-pay setting, to which prior results do not apply.

  Close to our work, sequential Tullock contests study players who bid in turn: Gao et al.~\cite{gao2023equilibria} solve the two-player case for a general power-weight function and specifically observe entry deterrence, but only under all-pay costs and a single action per player, neither of which holds in transaction ordering.
  Hinnosaar~\cite{hinnosaar2024optimal} solves the multi-player case, but shares the same two assumptions and, in addition, applies only to a linear weight function.
  Our model instead goes beyond all-pay and lets a player act at multiple steps and publish multiple transactions (spam), both inherent to transaction ordering.

  %%%%%%%%%%%%%%%%%%%%%%%%%%%%%%%%%%%%%%%%%%%%%%%%%%%%%%%%%%%%%%%%%%%%%%%%%%%%%%%%
  %%%%%%%%%%%%%%%%%%%%%%%%%%%%%%%%%%%%%%%%%%%%%%%%%%%%%%%%%%%%%%%%%%%%%%%%%%%%%%%%
  %%%%%%%%%%%%%%%%%%%%%%%%%%%%%%%%%%%%%%%%%%%%%%%%%%%%%%%%%%%%%%%%%%%%%%%%%%%%%%%%

    \section{Model}\label{sec:model}

  We consider a \emph{regular user} and a set of \emph{attackers} competing on a blockchain~(\S\ref{sec:model:participants}) for a reward~$R$~(\S\ref{sec:model:execution}).
  Transaction publication induces a causality relation among transactions~(\S\ref{sec:model:system}).
  Our goal is to design a protocol that achieves Causal Ordering, Anti-Spam, and Profitability~(\S\ref{sec:model:goal}).

  %%%%%%%%%%%%%%%%%%%%%%%%%%%%%%%%%%%%%%%%%%%%%%%%%%%%%%%%%%%%%%%%%%%%%%%%%%%%%%%%
  %%%%%%%%%%%%%%%%%%%%%%%%%%%%%%%%%%%%%%%%%%%%%%%%%%%%%%%%%%%%%%%%%%%%%%%%%%%%%%%%

      \subsection{Participants and Blockchain}\label{sec:model:participants}
  The system comprises rational participants $\psets = \{U\} \cup \mathbf{A}$: a user~$U$ and a set of attackers~$\mathbf{A}$.
  They all take actions to maximize their own revenue.
  Nodes called validators run a state machine replication protocol to maintain a \emph{blockchain}.

  The blockchain is a stateful virtual machine that executes a totally ordered sequence of \emph{transactions}.
  Each transaction is a log entry containing an instruction that updates the state.
  Participants create and publish these transactions.
  Formally, a transaction~$\tx$ is a tuple~$(\txcontent, \textit{bid})$, where $\txcontent$ includes the instruction and the creator, and \textit{bid} is the fee that the creator is willing to pay for the transaction to be included on-chain.
  We denote the \emph{bid} of a transaction~$\tx$ by~${b(\tx) > 0}$, abbreviated~$b$ when~$\tx$ is clear from context.
  Once the transaction is on chain, the blockchain charges the creator the bid (or a ratio thereof).

  Validators maintain the blockchain by running an \emph{ordering protocol} that extends the transaction sequence with new transactions; the virtual machine then executes the transactions in the order the protocol determines.
  We abstract away the blockchain implementation details and treat the ordering protocol as an idealized service that receives transactions from participants and processes them.

  Each execution of the ordering protocol consists of two phases: a \emph{collection phase} and an \emph{ordering phase}.
  In the collection phase, the protocol gathers the set of transactions it has received from participants, which we denote by~$\alltx$.
  In the ordering phase, the protocol runs an ordering function~$f$ on~$\alltx$.
  We assume an external random beacon pulse~$Q$ that is unpredictable before the ordering phase and becomes available when it begins; such beacons are well-established~\cite{das2022spurt,bhat2021randpiper} and are used in blockchain mechanism design~\cite{chung2023foundations,chen2025prrr}.
  Formally, $f$ takes~$\alltx$ and~$Q$ as input and outputs an ordered vector~$f(\alltx, Q)$.
  The blockchain appends this output to its transaction sequence; the current execution then ends and the next begins.

  %%%%%%%%%%%%%%%%%%%%%%%%%%%%%%%%%%%%%%%%%%%%%%%%%%%%%%%%%%%%%%%%%%%%%%%%%%%%%%%%
  %%%%%%%%%%%%%%%%%%%%%%%%%%%%%%%%%%%%%%%%%%%%%%%%%%%%%%%%%%%%%%%%%%%%%%%%%%%%%%%%

      \subsection{Reward}\label{sec:model:execution}
  The state of the system allows any participant to earn an on-chain reward~$R$ by publishing a transaction.
  Only the creator of the first transaction for it in the blockchain's transaction order earns~$R$; every other transaction earns nothing.
  For simplicity, we ignore transactions that do not compete for~$R$.

  The user knows of the reward at time~$0$ and can publish a transaction to receive it throughout the window~$[0, T]$, the collection phase of a single execution of the ordering protocol.
  An attacker is initially unaware of the reward; only after the user publishes a transaction can the attacker observe it and publish her own to compete for the reward.

  We now describe the payment for transactions in~$\alltx$.
  We call the first transaction in~$f(\alltx, Q)$ the \emph{winner} and denote it by~$\winner(\alltx, Q)$.
  For any transaction~${\tx \in \alltx}$, if it is the winner, i.e.~${\tx = \winner(\alltx, Q)}$, its creator earns a reward~$R$ and pays the full bid~$b(\tx)$;
  otherwise, the creator pays a ratio~$\gamma > 0$ of the bid, the \emph{losing-fee rate}~$\gamma$, so the payment is~$\gamma\, b(\tx)$.
  This rate reflects that, in practical blockchains, a non-winning transaction reverts along a different execution path and pays only for the resources it consumes~\cite{daian2020flash}, usually less than the winner's payment but occasionally even slightly more~(\S\ref{app:practical:gamma}).
  The rate~$\gamma$ varies across competitions, but a participant knows it accurately in advance by locally pre-executing her transaction.

  %%%%%%%%%%%%%%%%%%%%%%%%%%%%%%%%%%%%%%%%%%%%%%%%%%%%%%%%%%%%%%%%%%%%%%%%%%%%%%%%
  %%%%%%%%%%%%%%%%%%%%%%%%%%%%%%%%%%%%%%%%%%%%%%%%%%%%%%%%%%%%%%%%%%%%%%%%%%%%%%%%
      \subsection{System Progress}\label{sec:model:system}

  An execution unfolds over the window~$[0, T]$.
  To compete for the reward, a participant either publishes new transactions or increases the bids of those she has already published, and she may do so at any time.
  She cannot decrease the bid of an existing transaction~\cite{daian2020flash}.

  A transaction a participant publishes becomes visible to the others only after some latency, as it propagates through the network, and this latency can be arbitrary.
  They observe it together with its bid.
  This is what lets an attacker react: The user may act at any time, whereas an attacker, initially unaware of the reward, publishes only after a user's first transaction becomes visible to her.

  We assume that the ordering protocol is censorship-resistant: Every transaction published during~$[0, T]$ enters~$\alltx$.
  This property is satisfied by a range of protocols that aggregate transactions across validators (e.g.,~\cite{kelkar2020order,kursawe2020wendy,kelkar2023themis,zhang2020byzantine}).
  Such aggregation prevents any single node from omitting or deferring transactions~\cite{chen2025prrr}.
  Among all transactions in~$\alltx$, denote by~$\txset_U$ the user's transactions and by~$\txset_A$ those of any attacker~${A \in \mathbf{A}}$. 

  Denote by~$\pubprotocol$ the \emph{publication strategy} of a participant, which specifies how she publishes transactions and updates bids based on what she has observed at any time.

  %%%%%%%%%%%%%%%%%%%%%%%%%%%%%%%%%%%%%%%%%%%%%%%%%%%%%%%%%%%%%%%%%%%%%%%%%%%%%%%%
  %%%%%%%%%%%%%%%%%%%%%%%%%%%%%%%%%%%%%%%%%%%%%%%%%%%%%%%%%%%%%%%%%%%%%%%%%%%%%%%%
      \subsection{Goal}\label{sec:model:goal}

  Our goal is to design a protocol, an ordering function~$f$ together with a publication strategy~$\pubprotocol$, i.e., a pair~$(f, \pubprotocol)$, that eliminates front-running, captured by three properties: \emph{Causal Ordering}, \emph{Anti-Spam}, and \emph{Profitability}.

  First, Causal Ordering requires that the protocol never makes any attacker's transaction the winner: Since every attacker publishes only after observing the user's transaction, the winner must be one of the user's transactions.
  The difficulty is that the ordering function~$f$ cannot observe this causality: From~$\alltx$ and~$Q$ alone it cannot tell which transaction was published in reaction to another, so it cannot single out the attacker's transaction and must secure the guarantee blindly.
  \begin{definition}[Causal Ordering]\label{def:causal}
    A protocol~$(f, \pubprotocol)$ satisfies \emph{Causal Ordering} if, for any random beacon pulse~$Q$, it never orders any attacker's transaction first, so no attacker front-runs the user successfully.
    Formally, let~$\winner(\alltx, Q)$ be the winner of the ordering~$f(\alltx, Q)$ and~$\txset_A$ the set of transactions an attacker~$A \in \mathbf{A}$ publishes; we require
    \begin{equation*}
      \forall Q, \ \winner(\alltx, Q) \notin \bigcup_{A \in \mathbf{A}} \txset_A.
    \end{equation*}
  \end{definition}

  Second, because $f$ orders transactions using the random beacon~$Q$, such randomization might invite \emph{spam}: A participant publishes many transactions to raise the chance that one of them is ordered first, reducing system efficiency.
  We therefore require that each participant publishes at most one transaction:
  \begin{definition}[Anti-Spam]\label{def:antispam}
    A protocol~$(f, \pubprotocol)$ satisfies \emph{Anti-Spam} if the user and every attacker publish at most one transaction, i.e., $|\txset_U| \le 1$ and for all $A \in \mathbf{A}$, $|\txset_A| \le 1$.
  \end{definition}

  Finally, a user who follows the prescribed publication strategy~$\pubprotocol$ should profit;
  otherwise she stops exploring the system for opportunities (e.g., arbitrage), eroding system efficiency.
  \begin{definition}[Profitability]\label{def:profitability}
    A protocol~$(f, \pubprotocol)$ satisfies \emph{Profitability} if the user obtains a strictly positive expected revenue when she follows the protocol's publication strategy~$\pubprotocol$.
  \end{definition}

  % Regular users and attackers may have incentives to deviate from~$\pubprotocol$ to raise their winning chances --- e.g., by spamming multiple transactions or launching a Concorde attack.
  % Our design must ensure that no such deviation pays off:
  % \begin{definition}[Incentive Compatibility]
  %   For every participant~$P_i \in \psets$, executing~$\pubprotocol_i$ is a dominant strategy.
  % \end{definition}

  %%%%%%%%%%%%%%%%%%%%%%%%%%%%%%%%%%%%%%%%%%%%%%%%%%%%%%%%%%%%%%%%%%%%%%%%%%%%%%%%
  %%%%%%%%%%%%%%%%%%%%%%%%%%%%%%%%%%%%%%%%%%%%%%%%%%%%%%%%%%%%%%%%%%%%%%%%%%%%%%%%
  %%%%%%%%%%%%%%%%%%%%%%%%%%%%%%%%%%%%%%%%%%%%%%%%%%%%%%%%%%%%%%%%%%%%%%%%%%%%%%%%

    \section{PRECEDE}\label{sec:precede}

  We present \emph{PRECEDE}, \ul{P}ower-weighted \ul{R}an\-dom\-i\-za\-tion for \ul{E}n\-forc\-ing \ul{C}au\-sal ordering through \ul{E}n\-try \ul{DE}ter\-rence, to address front-running.
  We provide the design intuition~(\S\ref{sec:precede:intuition}), the ordering protocol~(\S\ref{sec:precede:ordering}), and the publication strategy by which a user deters entry, achieving Causal Ordering (\S\ref{sec:publication:deterrence}).
  We then show that the protocol satisfies Anti-Spam~(\S\ref{sec:precede:spam}) and Profitability~(\S\ref{sec:precede:profit}).

  %%%%%%%%%%%%%%%%%%%%%%%%%%%%%%%%%%%%%%%%%%%%%%%%%%%%%%%%%%%%%%%%%%%%%%%%%%%%%%%%
  %%%%%%%%%%%%%%%%%%%%%%%%%%%%%%%%%%%%%%%%%%%%%%%%%%%%%%%%%%%%%%%%%%%%%%%%%%%%%%%%
      \subsection{Design Intuition}\label{sec:precede:intuition}

  PRECEDE orders transactions solely by the bid each carries, which is on-chain and verifiable, rather than by timestamps or arrival times.
  Concretely, the protocol fixes an exponent $k > 1$ and assigns each transaction $\tx$ the weight $w(b(\tx)) = b^k(\tx)$, ordering it first with probability proportional to its weight.
  Our design principle is \emph{entry deterrence}: The user publishes a bid high enough so that the attacker cannot profit by \emph{entry}, i.e., by publishing a transaction to compete, so the attacker abstains.

  Whether the user can deter the attacker while still profiting hinges on the exponent $k$.
  As $k$ grows to infinity, the highest bid wins with probability approaching one.
  The design then degenerates into the deterministic descending-bid ordering of deployed blockchains~\cite{buterin2013ethereum,blackshear2024sui,yakovenko2018solana}.
  Under this rule the user can deter entry only by bidding the entire reward $R$: With any smaller bid, the attacker will bid slightly higher, win with certainty, and earn a positive revenue.
  Such deterrence leaves the user no profit, violating Profitability.

  A carefully chosen $k$ softens this rule, letting the user deter entry more cheaply.
  Overbidding no longer guarantees a win, and a non-winning counter-bid pays the ratio $\gamma$ of its bid.
  The attacker thus faces a dilemma: A small counter-bid is cheap but, against the user's large weight, wins with too low a probability to profit, while a large one wins with high probability but, since $k > 1$, must be so large that its cost exceeds the reward $R$.
  Neither option yields a positive expected revenue, so a bid strictly below $R$ suffices to deter the attacker.

  At the other extreme, $k = 0$ makes every bid carry the same weight, so the ordering degenerates into uniform random selection.
  A participant's winning probability then grows with the number of transactions she publishes, inviting spam~\cite{wang2026blockspace}, violating the Anti-Spam property.
  An exponent $k > 1$ makes \emph{merging}, placing the whole bid on a single transaction, attractive on the winning rate:
  Because $b^k$ is superadditive, a single bid outweighs the same amount split across several transactions, so merging gives a higher winning probability than spamming.
  However, merging is not free.
  When the losing-fee rate $\gamma$ is smaller than $1$, as is common~\cite{daian2020flash}, a losing transaction pays only the fraction $\gamma$ of its bid, while the winner pays in full.
  Concentrating the whole bid on a single transaction thus charges this full amount whenever it wins, raising the expected payment.
  As before, a carefully chosen $k$ tips this balance: It should be large enough that the winning-probability gain from merging outweighs the extra payment, so each participant publishes a single transaction.

  % %%%%%%%%%%%%%%%%%%%%%%%%%%%%%%%%%%%%%%%%%%%%%%%%%%%%%%%%%%%%%%%%%%%%%%%%%%%%%%%%
  % %%%%%%%%%%%%%%%%%%%%%%%%%%%%%%%%%%%%%%%%%%%%%%%%%%%%%%%%%%%%%%%%%%%%%%%%%%%%%%%%
  %     \subsection{Publication Strategy}\label{sec:precede:publication}

  %%%%%%%%%%%%%%%%%%%%%%%%%%%%%%%%%%%%%%%%%%%%%%%%%%%%%%%%%%%%%%%%%%%%%%%%%%%%%%%%
  %%%%%%%%%%%%%%%%%%%%%%%%%%%%%%%%%%%%%%%%%%%%%%%%%%%%%%%%%%%%%%%%%%%%%%%%%%%%%%%%
      \subsection{Ordering Protocol}\label{sec:precede:ordering}

\begin{algorithm}[t]
  \caption{PRECEDE Ordering Function}\label{alg:pwro}
  \begin{spacing}{1.2}
  \begin{algorithmic}[1]
    \Require{Transaction set $\alltx$, random beacon $Q$, exponent parameter $k$}
    \Ensure{Ordered sequence of transactions}
    \ForAll{$\tx \in \alltx$ \textbf{in parallel}}
      \State $r(\tx) \gets H(\tx \mathbin{\|} Q)$ \Comment{Random value in $[0,1]$}
      \State $s(\tx) \gets r(\tx)^{1/{b^k(\tx)}}$ \Comment{set $s(\tx) = 0$ if $b(\tx) = 0$}
    \EndFor
    \State $\pi \gets$ sequence of transactions in $\alltx$ sorted by $s(\cdot)$ in descending order
    \State \Return{$\pi$}
  \end{algorithmic}
  \end{spacing}
\end{algorithm}

  Recall that an ordering protocol runs in two phases: a collection phase, in which it gathers every transaction it receives into the set $\alltx$, and an ordering phase, in which it applies an \emph{ordering function} $f$ to $\alltx$ and the random beacon $Q$ to produce the order $f(\alltx, Q)$.

  PRECEDE's ordering function $f$ realizes a \emph{weighted ordering}: It gives each transaction a weight and orders one transaction ahead of another with probability proportional to its weight.
  PRECEDE uses a power-weight function, where $\tx$ has weight~$w(b(\tx)) = b^k(\tx)$, and $k>1$ is the \emph{exponent parameter} of the protocol.
  For two transactions $\tx_1$ and~$\tx_2$, we say $\tx_1 \prec_f^Q \tx_2$ if $f$ orders~$\tx_1$ before~$\tx_2$ given~$Q$.
  By the definition of a weighted ordering, 
  \begin{equation}\label{eq:pwro:pairwise}
    \Pr[\tx_1 \prec_f^Q \tx_2] = \frac{w(b(\tx_1))}{w(b(\tx_1)) + w(b(\tx_2))}.
  \end{equation}

  However, to serve as an ordering function, $f$ should also satisfy three further properties.
  First, $f$ must produce a total order, not just the pairwise outcomes above: Realizing the pairwise probabilities independently could create a cycle such as $\tx_1 \prec_f^Q \tx_2 \prec_f^Q \tx_3 \prec_f^Q \tx_1$, which no ordering satisfies.
  Second, $\alltx$ holds more than the transactions competing for one reward; it also contains transactions competing for other rewards, and transactions that do not compete at all, and these must not affect the order among the transactions competing for the same reward.
  Third, $f$ must be efficient to compute.

  We meet all three with a score-based method~\cite{efraimidis2006weighted} that scores all transactions in parallel and then sorts once.
  Denote by $H(\cdot)$ a random oracle that maps any input to a value uniformly distributed in the range $[0,1]$ and by~$\mathbin{\|}$ the concatenation operator.
  For each transaction $\tx \in \alltx$ we draw a random value ${r(\tx) = H(\tx \mathbin{\|} Q)}$ and set its score ${s(\tx) = r(\tx)^{1/{w({b(\tx)})}}}$.
  The ordering function orders the transactions by descending score: $\tx_1 \prec_f^Q \tx_2$ if and only if $s(\tx_1) > s(\tx_2)$.
  Algorithm~\ref{alg:pwro} states $f$ in full.

  This method realizes the pairwise probability (Equation~\ref{eq:pwro:pairwise})~\cite{efraimidis2006weighted}, and it meets the three properties.
  Sorting by score yields a single total order, and scoring in parallel before sorting once is efficient.
  Moreover, PRECEDE scores each transaction by $s(\tx)$, which depends only on the bid of transaction $\tx$ and the beacon $Q$, not on any other transaction.
  Hence the relative order among the transactions competing for one reward is independent of all other transactions, exactly as if they were ordered alone, so each competition can be analyzed in isolation.
  Restricting our model to the competition for a single reward is therefore without loss of generality.

  %%%%%%%%%%%%%%%%%%%%%%%%%%%%%%%%%%%%%%%%%%%%%%%%%%%%%%%%%%%%%%%%%%%%%%%%%%%%%%%%
  %%%%%%%%%%%%%%%%%%%%%%%%%%%%%%%%%%%%%%%%%%%%%%%%%%%%%%%%%%%%%%%%%%%%%%%%%%%%%%%%
      \subsection{Deterrence Strategy}\label{sec:publication:deterrence}
  By definition~(\S\ref{sec:model:execution}), the user is the participant who first identifies the reward, and hence the first to publish; an attacker is any participant who later observes her bid and may enter to compete.
  We realize entry deterrence as a publication strategy: The user publishes a single bid large enough so no attacker can profit by competing.
  We compute the attacker's revenue~(\S\ref{sec:publication:deterrence:payoff}), derive the bid that deters entry~(\S\ref{sec:publication:deterrence:deterrence}), and provide the user's strategy based on it~(\S\ref{sec:publication:deterrence:protocol}).

  \subsubsection{Revenue of the attacker}\label{sec:publication:deterrence:payoff}
  We consider the case in which an attacker, who has not yet published any transaction, has observed a nonempty set of transactions~$\txset$ from the user, and now publishes a single transaction to compete for the reward.
  Denote the total weight of the set~$\txset$ by:
  \begin{equation*}
  W = \sum_{\tx \in \txset} b^k(\tx) > 0 .
  \end{equation*}
  The attacker considers entering with a single transaction $\tx_A$ of bid $b(\tx_A) > 0$, hereafter simply $b$.
  Her transaction has weight $b^{k}$, so PRECEDE places it first with probability ${q = b^{k}/(b^{k} + W)}$ (Equation~\ref{eq:pwro:pairwise}).
  Under the payment rule, winning yields $R - b$ (collect $R$, pay the full bid $b$) and losing yields $-\gamma b$. 
  Her expected revenue is therefore
  \begin{align}\label{eq:payoff:u}
      u(b; W)
      &\;=\; q\,(R - b) \;-\; (1 - q)\,\gamma\,b \nonumber\\
      &\;=\; \frac{b^{k}(R - b) \;-\; \gamma\,b\,W}{b^{k} + W} \nonumber\\
      &\;=\; \frac{b\,\bigl(R\,b^{k-1} - b^{k} - \gamma\,W\bigr)}{b^{k} + W}.
  \end{align}

  \subsubsection{The deterrence weight and bid}\label{sec:publication:deterrence:deterrence}

  We now find the minimal weight $W$ of observed transactions that deters the attacker from entry.
  The attacker earns revenue $u(b; W)$ from a bid $b > 0$ and $0$ from abstaining, so entry is unprofitable exactly when $u(b; W) \le 0$ for all $b > 0$. 
  The following lemma gives the exact threshold of $W$ at which this condition holds.
  \begin{lemma}[Deterrence weight]\label{lem:2sgame:deterrence}
    Fix the exponent parameter $k > 1$ and the losing-fee rate $\gamma > 0$, and let $W$ be the total weight of transactions the attacker observes.
    Then the attacker's expected revenue with any single bid $b > 0$ satisfies $ u(b; W) \le 0 $ if and only if
    \begin{equation*}
        W \;\ge\; \frac{(k-1)^{k-1}}{\gamma\,k^{k}}\,R^{k}.
    \end{equation*}
  \end{lemma}

  Intuitively, since the attacker's revenue $u(b; W)$ is positive if and only if $R\,b^{k-1} - b^{k} - \gamma\,W > 0$ (Equation~\ref{eq:payoff:u}), we then calculate the minimum $W$ that makes this expression non-positive for all $b > 0$.
  \begin{showWhenSubmit}
    The proof is in the full version of this report~\cite{artifact}.
  \end{showWhenSubmit}
  \proofloc{app:proof:2sgame-deterrence}

  We name the threshold of the weight $W$ that deters entry the \emph{deterrence weight}, denoted:
  \begin{equation*}
    W^{k,\gamma}_{\dtr} \;\coloneqq\; \frac{(k-1)^{k-1}}{\gamma\,k^{k}}\,R^{k} .
  \end{equation*}

  By Lemma~\ref{lem:2sgame:deterrence}, a user who wants to deter all entry with a single bid must give that bid weight at least $W^{k,\gamma}_{\dtr}$. The smallest such bid, which we call the \emph{deterrence bid} $b^{k,\gamma}_{\dtr}$, has weight exactly $W^{k,\gamma}_{\dtr}$:
  \begin{equation}\label{eq:2sgame:b1-det}
    b^{k,\gamma}_{\dtr} \;\coloneqq\; \bigl(W^{k,\gamma}_{\dtr}\bigr)^{1/k} \;=\; R\left(\frac{(k-1)^{k-1}}{\gamma\,k^{k}}\right)^{1/k} .
  \end{equation}

  This bid achieves entry deterrence by making any attacker's transaction unprofitable.
  \begin{theorem}[Entry deterrence]\label{thm:2sgame:deterrence}
    Fix the exponent parameter $k > 1$ and the losing-fee rate $\gamma > 0$, and suppose the user publishes the deterrence bid $b^{k,\gamma}_{\dtr}$.
    Then the attacker's expected revenue with any single bid $b > 0$ is non-positive, i.e., $u(b; W^{k,\gamma}_{\dtr}) \le 0$.
  \end{theorem}
  This follows directly from Lemma~\ref{lem:2sgame:deterrence}: The deterrence bid has weight $\bigl(b^{k,\gamma}_{\dtr}\bigr)^{k} = W^{k,\gamma}_{\dtr}$ (Equation~\ref{eq:2sgame:b1-det}), so the attacker observes total weight $W = W^{k,\gamma}_{\dtr}$, at which $u(b; W^{k,\gamma}_{\dtr}) \le 0$ for all $b > 0$.

  \subsubsection{Strategy}\label{sec:publication:deterrence:protocol}

  We give a single strategy that every participant $P$ runs; what she observes, not her role, determines her action.
  The bound of Lemma~\ref{lem:2sgame:deterrence} depends only on the weight~$P$ observes, not on who produced it, so one deterrence bid deters every later participant.
  Therefore, the strategy prescribes $P$'s action in three cases.
  If she has not yet observed any transaction, meaning she is the first to publish a transaction to her knowledge, she publishes the deterrence bid $b^{k,\gamma}_{\dtr}$.
  If instead she observes a set of transactions whose total weight already deters entry, that is $W \ge W^{k,\gamma}_{\dtr}$, she publishes nothing.
  In the remaining case, $P$ publishes nothing; the user, being the first to publish, never reaches this case, so the choice is immaterial and may be replaced by any other action.
  Algorithm~\ref{alg:pub} states this strategy $\pubprotocol$, which $P$ runs once, at the first time $t \in [0, T]$ she learns of the reward $R$.
  At every other time in the window she does nothing, leaving any bid she has published unchanged.

  \begin{algorithm}[t]
    \caption{Deterrence Publication Strategy $\pubprotocol$ for $P$}\label{alg:pub}
    \begin{spacing}{1.2}
    \begin{algorithmic}[1]
      \Require{reward $R$, losing-fee rate $\gamma$, exponent parameter $k$, the set $\txset$ of transactions $P$ has received}
      \State $W \gets \sum_{\tx \in \txset} b^{k}(\tx)$ \Comment{total weight of transactions}
      \If{$\txset = \emptyset$}  \Comment{$P$ has not observed any transaction}
        \State publish a transaction with bid
        \Statex \hspace{\algorithmicindent}$b^{k,\gamma}_{\dtr} = R\bigl(\tfrac{(k-1)^{k-1}}{\gamma\,k^{k}}\bigr)^{1/k}$ \Statex \Comment{publish the deterrence bid}
      \ElsIf{$W \ge W^{k,\gamma}_{\dtr}$}\Comment{entry is unprofitable}
        \State publish nothing
      \Else\Comment{the choice here is immaterial}
        \State publish nothing
      \EndIf
    \end{algorithmic}
    \end{spacing}
  \end{algorithm}

  %%%%%%%%%%%%%%%%%%%%%%%%%%%%%%%%%%%%%%%%%%%%%%%%%%%%%%%%%%%%%%%%%%%%%%%%%%%%%%%%
  %%%%%%%%%%%%%%%%%%%%%%%%%%%%%%%%%%%%%%%%%%%%%%%%%%%%%%%%%%%%%%%%%%%%%%%%%%%%%%%%
      \subsection{Anti-Spam Requirements}\label{sec:precede:spam}

  We show that PRECEDE satisfies Anti-Spam.
  We consider any number of participants and replace one participant's transactions by a single bid of the same weight~(\S\ref{sec:precede:spam:setup}). 
  We analyze the resulting revenues~(\S\ref{sec:precede:spam:revenue}), show that the replacement leaves every other participant's revenue unchanged~(\S\ref{sec:precede:spam:others}), and show that the replaced participant is not worse off if $k$ is no less than a threshold that depends on $\gamma$~(\S\ref{sec:precede:spam:bound}).

  \subsubsection{Setup}\label{sec:precede:spam:setup}

  Anti-Spam is a per-participant property: Whether a rational participant gains by publishing multiple transactions does not depend on her role as user or attacker.
  We therefore drop this distinction and consider a general setting with $n$ participants $\psets = \{P_1, \ldots, P_n\}$, which is both more general, covering any number of participants, and stronger, since the guarantee holds for every one of them.
  By time $T$, $\alltx$ collects, for all $1 \leq i \leq n$, $P_i$'s set of~$m_i$ transactions with their final bids; we denote the $m_i$ bids of~$P_i$ by a vector $\bidvector_i = (b_i^{(1)}, \ldots, b_i^{(m_i)})$ of length $m_i \geq 0$, with $b_i^{(j)} > 0$ for all $1 \leq j \leq m_i$.
  The empty vector $\bidvector_i = ()$ denotes abstention.

  Without loss of generality, we assume $P_1$ publishes more than one transaction, i.e., $m_1 > 1$; the total weight of her transactions is $W_1 = \sum_{j=1}^{m_1} (b_1^{(j)})^k$.
  Denote by $\bar b_1$ the single bid that has the same total weight:
  \begin{equation}\label{eq:sybil:barb}
  \bar b_1
  \coloneqq
  \left(\sum_{j=1}^{m_1}(b_1^{(j)})^k\right)^{1/k}.
  \end{equation}

  \subsubsection{Revenue}\label{sec:precede:spam:revenue}

  We analyze each participant's revenue in two scenarios: the original one, in which $P_1$ publishes~$\bidvector_1$, and an alternative one, which is identical except that $P_1$ publishes the single bid $\bar b_1$ instead. 
  We detail the notation for the original scenario; the alternative scenario is obtained by substituting $(\bar b_1)$ for $\bidvector_1$ throughout.
  The participants' bids excluding $P_1$'s, $\bidvector_{-1} = (\bidvector_2, \ldots, \bidvector_n)$, are fixed.
  Denote by~$W(\bidvector_1, \bidvector_{-1})$ the total weight of all transactions:
  \begin{equation*}
    W(\bidvector_1, \bidvector_{-1})\;=\; \sum_{i=1}^n \sum_{j=1}^{m_i} (b_i^{(j)})^k .
  \end{equation*}
  Since the replacement bid $\bar b_1$ has the same weight as $\bidvector_1$ (Equation~\ref{eq:sybil:barb}), the total weight is the same in both scenarios:
  \begin{equation*}
    W(\bidvector_1, \bidvector_{-1}) = W((\bar b_1), \bidvector_{-1}).
  \end{equation*}

  Denote by $q(b, \bidvector_{-1}, \bidvector_1)$ the probability that a transaction of bid $b$ is ordered first:
  \begin{equation}\label{eq:sybil:q}
    q(b, \bidvector_{-1}, \bidvector_1) \;=\; \frac{b^k}{W(\bidvector_1, \bidvector_{-1})}.
  \end{equation}

  Denote by $C_i(\bidvector_1, \bidvector_{-1})$ the expected total payment of participant $P_i$.
  Recall that a transaction pays the ratio $\gamma$ of its bid whether or not it is ordered first, and the one ordered first pays its bid in full, that is, an additional $(1-\gamma)$ of it. 
  So a transaction of bid $b_i^{(j)}$ pays $\gamma\,b_i^{(j)}$ for sure, plus the extra~$(1-\gamma)\,b_i^{(j)}$ with probability $q(b_i^{(j)}, \bidvector_{-1}, \bidvector_1)$. Summing this expected payment over all of $P_i$'s transactions,
  \begin{equation}\label{eq:sybil:Ci}
    C_i(\bidvector_1, \bidvector_{-1}) \;=\; \sum_{j=1}^{m_i}\Bigl(\gamma\,b_i^{(j)} \;+\; (1-\gamma)\,q(b_i^{(j)}, \bidvector_{-1}, \bidvector_1)\,b_i^{(j)}\Bigr).
  \end{equation}

  The revenue $u_i(\bidvector_1, \bidvector_{-1})$ of $P_i$ is the expected reward from winning minus the expected payment
  \begin{equation}\label{eq:sybil:utility}
    u_i(\bidvector_1, \bidvector_{-1}) \;=\; R\sum_{j=1}^{m_i} q(b_i^{(j)}, \bidvector_{-1}, \bidvector_1) \;-\; C_i(\bidvector_1, \bidvector_{-1}).
  \end{equation}

\subsubsection{Other participants are unaffected}\label{sec:precede:spam:others}
  Replacing $P_1$'s bid vector $\bidvector_1$ with the single bid $(\bar b_1)$ preserves the total weight, since $\bar b_1^k = \sum_{j=1}^{m_1}(b_1^{(j)})^k$ (Equation~\ref{eq:sybil:barb}). Every other participant's winning probability and payment depend on~$\bidvector_1$ only through the total weight (Equations~\ref{eq:sybil:q}--\ref{eq:sybil:utility}), so both, and hence her revenue, are unchanged, implying the following proposition.
  \begin{proposition}[Merging is externality-free]\label{cor:sybil:others}
    For every $i \neq 1$, replacing $P_1$'s bid vector $\bidvector_1$ with the single bid $(\bar b_1)$ leaves participant $P_i$'s revenue unchanged,
    \begin{equation*}
      u_i\bigl((\bar b_1), \bidvector_{-1}\bigr) = u_i(\bidvector_1, \bidvector_{-1}).
    \end{equation*}
  \end{proposition}

\subsubsection{Lower bound for the exponent $k$}\label{sec:precede:spam:bound}

Intuitively, the larger $k$ is, the more weight a higher bid carries, so the single bid $\bar b_1$ (Equation~\ref{eq:sybil:barb}) that matches the total weight of $P_1$'s bids can be much smaller than their sum.
Since each transaction pays a fee proportional to its bid, a large enough $k$ makes merging profitable for $P_1$.
We thereby provide a tight lower bound on $k$ that guarantees Anti-Spam.

\begin{theorem}[Anti-Spam Threshold]\label{thm:sybil:tight}
    Fix the exponent parameter $k>1$ and the losing-fee rate $\gamma>0$.
    For every bid vector $\bidvector_1$ of $P_1$ and every profile $\bidvector_{-1}$ of the other participants' bids, replacing $\bidvector_1$ with the single bid $(\bar b_1)$ does not decrease $P_1$'s revenue,
    \begin{equation*}
      u_1\bigl((\bar b_1), \bidvector_{-1}\bigr) \;\geq\; u_1(\bidvector_1, \bidvector_{-1}),
    \end{equation*}
    if and only if
    \begin{equation*}
      k \;\geq\; \frac{\ln 2}{\ln(1+\gamma)} .
    \end{equation*}
\end{theorem}

% \begin{theorem}[Anti-Spam Threshold]\label{thm:sybil:tight}
% Fix $\gamma\in(0,1]$ and $k>1$. For every nonempty bid vector $\bidvector_1$ of $P_1$ and every profile $\bidvector_{-1}$ of the other participants' bids, let
% \[
%     \bar b_1
%     =
%     \left(\sum_{j=1}^{m_1} (b_1^{(j)})^k\right)^{1/k}
% \]
% be the single bid with the same PRECEDE weight as $\bidvector_1$. Then replacing $\bidvector_1$ by $(\bar b_1)$ does not decrease $P_1$'s revenue,
% \begin{equation*}
%     u_1\bigl((\bar b_1),\bidvector_{-1}\bigr)
%     \ge
%     u_1(\bidvector_1,\bidvector_{-1}),
% \end{equation*}
% for all such $\bidvector_1$ and $\bidvector_{-1}$ if and only if
% \begin{equation}\label{eq:sybil:tight-threshold}
%     k
%     \ge
%     \frac{\ln 2}{\ln(1+\gamma)}.
% \end{equation}
% \end{theorem}

  To prove this, we note that replacing $P_1$'s bids with the single equal-weight bid leaves her winning probability, and hence her expected reward, unchanged, so it raises her revenue exactly when it lowers her expected cost.
  This is hardest when she uses two equal bids and no other weight competes, which gives the threshold $k \ge \ln 2/\ln(1+\gamma)$; being the hardest case, it makes the same threshold suffice for any number of bids. \proofloc{app:proof:sybil-tight}
  \begin{showWhenSubmit}
     The proof is in the full version of this report~\cite{artifact}.
  \end{showWhenSubmit}

  %%%%%%%%%%%%%%%%%%%%%%%%%%%%%%%%%%%%%%%%%%%%%%%%%%%%%%%%%%%%%%%%%%%%%%%%%%%%%%%%
  %%%%%%%%%%%%%%%%%%%%%%%%%%%%%%%%%%%%%%%%%%%%%%%%%%%%%%%%%%%%%%%%%%%%%%%%%%%%%%%%
      \subsection{Profitability of Deterrence}\label{sec:precede:profit}
  Profitability requires the user to obtain strictly positive expected revenue when following the deterrence strategy.
  Since the attacker would not compete, the user wins the reward $R$ and pays the deterrence bid~$b^{k,\gamma}_{\dtr}$.
  Denote by $u^{k,\gamma}_{\dtr}$ the user's revenue under the deterrence strategy:
  \begin{equation*}
    u^{k,\gamma}_{\dtr} \coloneqq R - b^{k,\gamma}_{\dtr}= R\left(1 - \left(\frac{(k-1)^{k-1}}{\gamma\,k^{k}}\right)^{1/k}\right).
  \end{equation*}

  We now show that whenever PRECEDE meets the Anti-Spam threshold of Theorem~\ref{thm:sybil:tight}, it also satisfies Profitability.

\begin{proposition}[Deterrence is profitable]\label{prop:det-profitable}
If
$
    k \ge \frac{\ln 2}{\ln(1+\gamma)},
$
then the deterrence strategy results in strictly positive revenue, namely
$
    u_{\dtr}^{k,\gamma}>0.
$
\end{proposition}

Intuitively, the deterrence strategy is profitable exactly when the deterrence bid is strictly smaller than the reward~$R$.
A larger exponent~$k$ pushes the deterrence bid below~$R$, so Profitability reduces to a lower bound on~$k$.
The Anti-Spam threshold from Theorem~\ref{thm:sybil:tight} imposes a stronger lower bound on $k$, and therefore automatically guarantees Profitability.
\begin{showWhenSubmit}
The proof is in the full version of this report~\cite{artifact}.
\end{showWhenSubmit}
\proofloc{app:proof:det-profitable}

  %%%%%%%%%%%%%%%%%%%%%%%%%%%%%%%%%%%%%%%%%%%%%%%%%%%%%%%%%%%%%%%%%%%%%%%%%%%%%%%%
  %%%%%%%%%%%%%%%%%%%%%%%%%%%%%%%%%%%%%%%%%%%%%%%%%%%%%%%%%%%%%%%%%%%%%%%%%%%%%%%%
  %%%%%%%%%%%%%%%%%%%%%%%%%%%%%%%%%%%%%%%%%%%%%%%%%%%%%%%%%%%%%%%%%%%%%%%%%%%%%%%%
  
      \section{Game-Theoretic Analysis}\label{sec:game}
  Our model gives rise to a sequential game played by a user and an attacker (\S\ref{sec:game:model}).
  We first solve a two-step version~(\S\ref{sec:game:2seq}), where the user leads and the attacker responds, proving that the user playing the deterrence strategy and the attacker abstaining form an SPNE; we then extend this result to the full game with an arbitrary number of steps~(\S\ref{sec:game:multistep}).

  %%%%%%%%%%%%%%%%%%%%%%%%%%%%%%%%%%%%%%%%%%%%%%%%%%%%%%%%%%%%%%%%%%%%%%%%%%%%%%%%
  %%%%%%%%%%%%%%%%%%%%%%%%%%%%%%%%%%%%%%%%%%%%%%%%%%%%%%%%%%%%%%%%%%%%%%%%%%%%%%%%
      \subsection{Game Model}\label{sec:game:model}

  The competition for reward gives rise to a sequential game.
  Throughout the game analysis, we assume ${k \ge \ln 2/\ln(1+\gamma)}$, which ensures both Anti-Spam and Profitability.
  We specify the game's players~(\S\ref{sec:game:model:players}), their actions~(\S\ref{sec:game:model:progress}), their utilities~(\S\ref{sec:game:model:utility}), and the solution concept we adopt~(\S\ref{sec:game:model:solution}).

  \subsubsection{Players}\label{sec:game:model:players}
  There are two players: a regular user $P_1$ and a single attacker $P_2$.
  Modeling a single attacker is not a reduction of our model with multiple attackers, but rather a worst-case assumption: It captures all attackers colluding under one shared utility, whereas independent attackers would also compete among themselves and make deterrence easier for the user.

  \subsubsection{Game progress and actions}\label{sec:game:model:progress}
 
  The two players compete during the collection window $[0,T]$, which we discretize into $2m$ steps (for any $m \geq 1$).
  The user $P_1$ and the attacker $P_2$ move in alternating turns: $P_1$ on the odd steps and $P_2$ on the even steps, with each player observing all earlier actions before acting.
  The number of steps is even so that the attacker moves last, capturing her advantage in observing and reacting to others' transactions, be it due to lower network latency or to her control over the set of placed transactions as the validator.

  At each of her steps, a player may publish new transactions or increase the bids of those she has already published~(\S\ref{sec:model:system}).
  By Theorem~\ref{thm:sybil:tight}, replacing any set of a player's transactions by a single bid of the same total weight never lowers her own revenue, so a rational player gains nothing by publishing more than one transaction.
  By Proposition~\ref{cor:sybil:others}, the same replacement leaves the other player's revenue unchanged.
  Without loss of generality, we therefore restrict our attention to strategies in which each player publishes at most a single transaction.

  Under this restriction, the action of $P_1$ on her $j$-th step (step $2j-1$) is a bid $b_1^{(j)} \in \mathbb{R}_{\geq 0}$, and that of $P_2$ on her $j$-th step (step $2j$) is a bid $b_2^{(j)} \in \mathbb{R}_{\geq 0}$, for $j = 1, \ldots, m$; a bid of $0$ denotes publishing no transaction.
  The bid is non-decreasing, $b_i^{(j)} \le b_i^{(j+1)}$, and although a player acts at every step of the game, a step with no bid increase ($b_i^{(j+1)} = b_i^{(j)}$) requires no action from her in practice.
  Since only~$P_1$ knows the reward initially, she moves first, and~${b_1^{(j)} = 0}$ forces $b_2^{(j)} = 0$: Without observing $P_1$'s transaction, $P_2$ cannot publish.
  We denote this $2m$-step game with parameters $m$, $k$, and $\gamma$ by $\mathcal{G}^{k,\gamma}_m$.

  \subsubsection{Utility}\label{sec:game:model:utility}
  The ordering protocol uses only each player's final bid $b_i^{(m)}$:
  It weighs a transaction by $w(b) = b^k$ and charges its creator on that same bid (Algorithm~\ref{alg:pwro}).
  Hence the winner, the payments, and each player's utility depend only on the two final-step bids $b_1^{(m)}$ and $b_2^{(m)}$, not on the intermediate steps.
  In particular, if $b_1^{(m)} = b_2^{(m)} = 0$, no transaction is published and both utilities are $0$.

  The total weight of the two transactions is
  \begin{equation*}
    W\bigl(b_1^{(m)}, b_2^{(m)}\bigr) \;\coloneqq\; \bigl(b_1^{(m)}\bigr)^k + \bigl(b_2^{(m)}\bigr)^k ,
  \end{equation*}
  and $P_i$ wins with probability
  \begin{equation}\label{eq:2sgame:winprob}
    q_i\bigl(b_1^{(m)}, b_2^{(m)}\bigr) \;\coloneqq\; \frac{\bigl(b_i^{(m)}\bigr)^k}{W\bigl(b_1^{(m)}, b_2^{(m)}\bigr)}
  \end{equation}
  when $W\bigl(b_1^{(m)}, b_2^{(m)}\bigr) > 0$, and $q_i\bigl(b_1^{(m)}, b_2^{(m)}\bigr) = 0$ otherwise.

  The utility of $P_i$ is her expected revenue. 
  Specializing the general formula in~Equation~\ref{eq:sybil:utility} to the two-player single-bid setting,
  \begin{equation}\label{eq:2sgame:utility}
    \begin{aligned}
      u_i\bigl(b_1^{(m)}, b_2^{(m)}\bigr) &\;=\; -\,\gamma\,b_i^{(m)} \;+\; R\,q_i \;-\; (1-\gamma)\,q_i\,b_i^{(m)} \\
          &\;=\; q_i\,\bigl(R - b_i^{(m)}\bigr) \;-\; \bigl(1 - q_i\bigr)\,\gamma\,b_i^{(m)} .
    \end{aligned}
  \end{equation}

  \subsubsection{Solution concept}\label{sec:game:model:solution}
  Our solution concept is the Subgame-Perfect Nash Equilibrium (SPNE): The players' strategies are a Nash equilibrium in every subgame.
  Here, a subgame is the game from some step $t$ onward, given the bids placed in the previous steps $1, \ldots, t-1$.
  % We compute the SPNE by backward induction from the attacker's last move at step~$2m$.

  %%%%%%%%%%%%%%%%%%%%%%%%%%%%%%%%%%%%%%%%%%%%%%%%%%%%%%%%%%%%%%%%%%%%%%%%%%%%%%%%
  %%%%%%%%%%%%%%%%%%%%%%%%%%%%%%%%%%%%%%%%%%%%%%%%%%%%%%%%%%%%%%%%%%%%%%%%%%%%%%%%
      \subsection{Two-Step Equilibrium}\label{sec:game:2seq}

  We first analyze a simple version of the game with only two steps ($m = 1$), so each player acts once and the game $\mathcal{G}^{k,\gamma}_1$ degenerates to a Stackelberg game~\cite{von2010market}.
  By backward induction, we derive the attacker's best response~(\S\ref{sec:game:2seq:p2}), then the user's~(\S\ref{sec:game:2seq:p1}), and combine them into an SPNE~(\S\ref{sec:game:2seq:eq}).

  \subsubsection{$P_2$'s best response}\label{sec:game:2seq:p2}
  If $b_1 = 0$, then $P_2$ cannot publish, so the best response is $b_2 = 0$.
  By Lemma~\ref{lem:2sgame:deterrence}, if $b_1 \ge b^{k,\gamma}_{\dtr}$, every bid earns $P_2$ non-positive utility, so $P_2$ abstains and the best response is $b_2 = 0$.

  For $0 < b_1 < b^{k,\gamma}_{\dtr}$, Lemma~\ref{lem:2sgame:deterrence} instead guarantees that some $b_2 > 0$ earns $P_2$ a strictly positive utility.
  Since~$u_2(b_1, \cdot)$ is continuous and tends to $-\infty$ as~${b_2 \to \infty}$,~$P_2$'s set of best responses
  \[
    \operatorname{BR}_2(b_1) \;\coloneqq\; \operatorname*{arg\,max}_{b_2 \ge 0} u_2(b_1, b_2)
  \]
  is nonempty and compact.
  This set may contain several bids, so we adopt the standard assumption of breaking ties in favor of a designated player~\cite{harsanyi1988general,dutting2019simple,collina2024efficient}, here~$P_1$: Among~$P_2$'s best responses we select the one maximizing~$P_1$'s utility.
  This choice is also natural in our game: A higher attacker bid only lowers~$P_1$'s utility~(Equation~\ref{eq:2sgame:utility}), so favoring~$P_1$ selects the attacker's lowest best-response bid:
  \[
    b_2^*(b_1) \;=\; \min \operatorname{BR}_2(b_1),
  \]
  which is well defined because $\operatorname{BR}_2(b_1)$ is nonempty and compact.
  Selecting the lowest bid also minimizes the attacker's effort: Her best responses all yield her the same utility, so the lowest one attains that utility at the smallest expected payment.

  In summary, $P_2$'s best response is
  \begin{equation*}
    b_2^*(b_1) \;=\;
    \begin{cases}
      0, & b_1 = 0 \ \text{ or } \ b_1 \ge b^{k,\gamma}_{\dtr}, \\[2pt]
      \min \operatorname{BR}_2(b_1), & 0 < b_1 < b^{k,\gamma}_{\dtr}.
    \end{cases}
  \end{equation*}

  \subsubsection{$P_1$'s best response}\label{sec:game:2seq:p1}
  Given $P_2$'s best response function~$b_2^*(b_1)$,~$P_1$ chooses the bid $b_1 \ge 0$ that maximizes her utility $u_1\bigl(b_1, b_2^*(b_1)\bigr)$.

  If $b_1 \ge b^{k,\gamma}_{\dtr}$, then $b_2^*(b_1) = 0$ and ${u_1(b_1, 0) = R - b_1}$, so~$P_1$'s utility is highest at the smallest deterring bid ${b_1 = b^{k,\gamma}_{\dtr}}$, where it equals $u^{k,\gamma}_{\dtr} > 0$ (Proposition~\ref{prop:det-profitable}).
  When $b_1 = 0$, her utility is $0 < u^{k,\gamma}_{\dtr}$, so abstention is never optimal for~$P_1$.

  When $0< b_1 < b^{k,\gamma}_{\dtr}$, $P_2$'s best response $b_2^*(b_1)$ is to publish a positive bid.
  We call this the \emph{accommodation strategy} of $P_1$.
  Denote by $u^{k,\gamma}_{\mathrm{acc}}$ the maximum utility $P_1$ can earn by the accommodation strategy assuming $P_2$ adopts her best response:
  \begin{equation}\label{eq:2sgame:u1-acc}
    u^{k,\gamma}_{\mathrm{acc}} \;\coloneqq\; \sup_{0 < b_1 < b^{k,\gamma}_{\dtr}} u_1\bigl(b_1, b_2^*(b_1)\bigr) .
  \end{equation}
  This supremum is finite, since $u_1 \le R$ for every bid profile. 
  Therefore, when $u^{k,\gamma}_{\dtr} > u^{k,\gamma}_{\mathrm{acc}}$, the only best response for $P_1$ is the deterrence strategy.

  Unlike the closed-form $u^{k,\gamma}_{\dtr}$, we do not have a closed-form expression for $u^{k,\gamma}_{\mathrm{acc}}$.
  Computing it requires the attacker's best response $b_2^*(b_1)$, the positive bid that maximizes~$P_2$'s utility at each $b_1$ in the accommodation range.
  Since it is positive, $b_2^*(b_1)$ is an interior optimum and thus satisfies the first-order condition $\partial u_2/\partial b_2 = 0$.
  This condition is transcendental in $b_2$ because the winning probability weighs bids by the power $b^k$~(Equation~\ref{eq:2sgame:winprob}), and for general~$k$ and $\gamma$ we cannot solve it in closed form, leaving us without a closed form for $b_2^*(b_1)$ or for $u^{k,\gamma}_{\mathrm{acc}}$.

  \subsubsection{Equilibrium}\label{sec:game:2seq:eq}
  We want \emph{deterrence} to be an equilibrium, with $P_1$ publishing the deterring bid $b^{k,\gamma}_{\dtr}$ and $P_2$ abstaining.
  The following theorem gives the condition under which this holds.

  \begin{theorem}[Deterrence equilibrium]\label{thm:2sgame:eq}
    Fix $\gamma > 0$ and ${k > 1}$ satisfying $k \ge \ln 2/\ln(1+\gamma)$.
    If $u^{k,\gamma}_{\dtr} > u^{k,\gamma}_{\mathrm{acc}}$, then $P_1$ publishing the deterring bid $b_1^* = b^{k,\gamma}_{\dtr}$ and $P_2$ abstaining, $b_2^*(b_1^*) = 0$, form an SPNE of the game~$\mathcal{G}^{k,\gamma}_1$.
  \end{theorem}

  The theorem follows directly from the best-response analysis: When $u^{k,\gamma}_{\dtr} > u^{k,\gamma}_{\mathrm{acc}}$, deterrence is $P_1$'s unique best response~(\S\ref{sec:game:2seq:p1}), to which $P_2$ responds by abstaining~(\S\ref{sec:game:2seq:p2}), so the profile is an SPNE.

  Theorem~\ref{thm:2sgame:eq} reduces the design problem to comparing~$u^{k,\gamma}_{\dtr}$ with $u^{k,\gamma}_{\mathrm{acc}}$. 
  The following lemma upper-bounds $u^{k,\gamma}_{\mathrm{acc}}$ by a closed-form expression.

  \begin{lemma}[Accommodation utility bound]\label{lem:game:acc-bound}
    For any $\gamma > 0$ and $k>1$, the accommodation utility $u^{k,\gamma}_{\mathrm{acc}}$ satisfies
    \begin{equation*}
      u^{k,\gamma}_{\mathrm{acc}} \;\le\; \frac{R}{k\,\min\{1,\gamma\}}.
    \end{equation*}
  \end{lemma}

  When $P_1$ accommodates, she wins with probability~$q_1^*$ and so earns $q_1^* R$ in expectation, ignoring the payment.
  We show that the attacker's optimal entry keeps~$q_1^*$ below $1/(k\,\min\{1,\gamma\})$, so the user's utility is below~$R/(k\,\min\{1,\gamma\})$.
  \begin{showWhenSubmit}
    The proof is in the full version of this report~\cite{artifact}.
  \end{showWhenSubmit}
  \proofloc{app:proof:acc-bound}

  The lemma above caps the accommodation utility at~$R/(k\,\min\{1,\gamma\})$.
  Deterrence is therefore the equilibrium whenever the deterrence utility exceeds this cap, which we show holds if $k$ is large enough.

  \begin{lemma}[Deterrence utility bound]\label{lem:game:explicit-k-payoff}
    Fix $\gamma > 0$. If
    \[
      k \;\geq\; \max\left\{2,\frac{\exp(1/\gamma)}{\gamma}\right\},
    \]
    then the utility of the deterrence strategy exceeds $\frac{R}{k\,\min\{1,\gamma\}}$:
    \begin{equation*}
      u^{k,\gamma}_{\dtr} \;>\; \frac{R}{k\,\min\{1,\gamma\}}.
    \end{equation*}
  \end{lemma}

  For $\gamma>1$, we bound the deterrence bid by its value at $\gamma=1$, leaving a utility above $R/k=R/(k\,\min\{1,\gamma\})$.
  For $0<\gamma\leq1$, we reduce the desired utility bound $R/(k\,\min\{1,\gamma\})=R/(k\gamma)$ to a one-variable inequality monotone in $k\gamma$, then check it at the smallest value permitted by the hypothesis, $k\gamma = \exp(1/\gamma)$.
  \begin{showWhenSubmit}
     The proof is in the full version of this report~\cite{artifact}.
  \end{showWhenSubmit}
  \proofloc{app:proof:explicit-k-payoff}

 Combining the two bounds gives a bound on~$k$: Once ${k \ge \max\{2,\exp(1/\gamma)/\gamma\}}$, Lemmas~\ref{lem:game:explicit-k-payoff} and~\ref{lem:game:acc-bound} imply ${u^{k,\gamma}_{\dtr} > R/(k\,\min\{1,\gamma\}) \ge u^{k,\gamma}_{\mathrm{acc}}}$.
  Therefore, Theorem~\ref{thm:2sgame:eq} makes deterrence the equilibrium.
  We record this explicit sufficient condition on~$k$ as the following corollary.

  \begin{corollary}\label{cor:game:explicit-k}
    Fix $\gamma > 0$. If $k\geq \max\{2,\exp(1/\gamma)/\gamma\}$, then~$P_1$ publishing the deterring bid $b_1^* = b^{k,\gamma}_{\dtr}$ and $P_2$ abstaining, $b_2^*(b_1^*) = 0$, form an SPNE of the game~$\mathcal{G}^{k,\gamma}_1$.
  \end{corollary}

  \subsection{Multiple-Step Equilibrium}\label{sec:game:multistep}

We now extend the result to the full $2m$-step game $\mathcal{G}^{k,\gamma}_m$ for any $m \ge 1$.
Intuitively, the two-step equilibrium should persist: $P_1$ publishes the deterring bid at her first step and never raises it, while $P_2$ never enters.
Because $P_2$ still moves last, any state at which $P_1$ has failed to deter reduces to the accommodation scenario of the two-step game; so, just as before, deterrence should be the equilibrium whenever its utility is no less than the accommodation utility.

However, there is still a gap in this argument.
SPNE demands an equilibrium in every subgame, that is, after every possible sequence of past bids, even those that rational play would never reach.
Standard backward induction would require solving these off-path subgames explicitly, which is prohibitively complicated.
Instead, we first show that each such subgame has an SPNE~(\S\ref{sec:game:multistep:existence}), then bound the utility~$P_1$ can obtain in any SPNE starting from her first positive bid~(\S\ref{sec:game:multistep:bound}), and finally extend the two-step result to $\mathcal{G}^{k,\gamma}_m$ by backward induction~(\S\ref{sec:game:multistep:extension}).

\subsubsection{Existence of SPNE}\label{sec:game:multistep:existence}

We first show the existence of an SPNE in any subgame where the user has already published a positive bid.
Fix a step $t\in\{1,\ldots,2m\}$.

A \emph{history} $h^{t-1}$ at the beginning of step $t$ is the bids players published in the earlier steps $1,\ldots,t-1$.
Such a history fixes each player's \emph{latest bid}, the last one she has placed.
For convenience of notation, we set ${b_1^{(0)}=b_2^{(0)}=0}$, meaning that neither player has yet published a bid before the first step.
We denote by $x_1(h^{t-1})$ the latest bid of $P_1$ and by $x_2(h^{t-1})$ that of~$P_2$:
\[
x_1(h^{t-1}) \coloneqq b_1^{(\lfloor t/2\rfloor)},
\qquad
x_2(h^{t-1}) \coloneqq b_2^{(\lfloor (t-1)/2\rfloor)}.
\]
The floor indices count how many times each player has acted before step $t$:
Before an odd step $2j-1$, both $P_1$ and~$P_2$ have acted $j-1$ times; before an
even step $2j$, $P_1$ has acted $j$ times and $P_2$ has still acted $j-1$ times.

We show that every subgame in which $P_1$'s latest bid is positive, $x_1(h^{t-1})>0$, has a pure-strategy SPNE.

\begin{lemma}[SPNE existence after a positive user bid]
\label{lem:game:positive-state-existence}
Consider a subgame starting from step $t$ with any history $h^{t-1}$ such that
\[
x_1(h^{t-1})=b_1^{(\lfloor t/2\rfloor)}>0.
\]
Then the subgame admits a pure-strategy SPNE.
\end{lemma}

By Hellwig and Leininger~\cite[Theorem~1]{hellwig1987existence}, a finite-horizon perfect-information game has a pure-strategy SPNE whenever every decision has a compact action set, the feasible-action correspondence is nonempty, compact-valued, and continuous in the history, and the utilities are continuous.
Our game satisfies these conditions except that its bid space is unbounded.
We therefore restrict the remaining bids to sufficiently large bounded intervals.
The lower endpoint of each interval is the acting player's latest bid, so the restriction respects the no-decrease constraint.
We then choose the upper bounds large enough so that every excluded bid is strictly worse than simply keeping the latest bid.
Thus an SPNE of the restricted compact game is also an SPNE of the original unbounded subgame.
\begin{showWhenSubmit}
    The proof is in the full version of this report~\cite{artifact}.
\end{showWhenSubmit}
\proofloc{app:proof:positive-state-existence}

\subsubsection{Bounding the utility}
\label{sec:game:multistep:bound}

We call a history at the beginning of an odd step an \emph{all-zero state} if both players' latest bids are zero, and a subgame \emph{all-zero} if it begins at an all-zero state.
We bound the utility $P_1$ can obtain by entering with a positive bid from an all-zero state.

\begin{lemma}[Positive-entry continuation bound]
\label{lem:game:positive-entry-bound}
Consider a subgame immediately after $P_1$ publishes a positive bid from an all-zero state, and let $\bigl(b_1^{(m)},b_2^{(m)}\bigr)$ be the final bids under any SPNE of this subgame.
Then
\[
u_1\bigl(b_1^{(m)},b_2^{(m)}\bigr)
\le
\max\left\{
u^{k,\gamma}_{\dtr},
u^{k,\gamma}_{\mathrm{acc}}
\right\}.
\]
\end{lemma}

In the two-step game, $P_1$ either deters $P_2$ and earns $u^{k,\gamma}_{\dtr}$, or accommodates her and earns at most $u^{k,\gamma}_{\mathrm{acc}}$.
The lemma asserts that the multi-step \emph{continuation}, the actions the players take in the subgame after $P_1$'s positive bid, cannot beat the larger of the two.
Since $P_1$'s utility depends only on the final bids $(b_1^{(m)},b_2^{(m)})$, the continuation is effectively a single two-step interaction, except that $P_2$'s final bid is bounded below by her earlier bid $b_2^{(m-1)}$.
This lower bound can only hurt $P_1$: Were it absent, $P_2$ would bid no higher, and a lower bid by $P_2$ only raises $P_1$'s winning probability and hence her utility whenever that utility is positive (and if it is already negative, the bound is immediate).
$P_1$'s utility is therefore at most that of the corresponding two-step outcome, which is in turn at most the deterrence or the accommodation utility.
\begin{showWhenSubmit}
    The proof is in the full version of this report~\cite{artifact}.
\end{showWhenSubmit}
\proofloc{app:proof:positive-entry-bound}

\subsubsection{Extension of the two-step SPNE}
\label{sec:game:multistep:extension}

We now extend the two-step deterrence equilibrium to the full $2m$-step game.
The above lemma caps the utility $P_1$ can gain by entering with a positive bid from an all-zero state.
With this bound in hand, we construct a deterrence equilibrium in every all-zero subgame by backward induction.

\begin{theorem}[Deterrence from every all-zero subgame]
\label{thm:game:multistep}
Suppose $u_{\dtr}^{k,\gamma}>u_{\mathrm{acc}}^{k,\gamma}$.
For every odd step $t=2j-1$, consider an all-zero subgame beginning at step $t$, that is,
\[
b_1^{(\lfloor t/2\rfloor)}
=
b_1^{(j-1)}
=
0
\]
and
\[
b_2^{(\lfloor (t-1)/2\rfloor)}
=
b_2^{(j-1)}
=
0.
\]
Then this all-zero subgame admits an SPNE in which
\[
b_1^{(\ell)}=b_{\dtr}^{k,\gamma}
\quad\text{and}\quad
b_2^{(\ell)}=0
\qquad
\text{for every }\ell=j,\ldots,m.
\]
\end{theorem}

We prove the result by backward induction over all-zero subgames. The last all-zero subgame is exactly the two-step game. For an earlier all-zero subgame, publishing the deterring bid gives $P_1$ revenue $u_{\dtr}^{k,\gamma}$, publishing zero moves the game to the next all-zero subgame, and the revenue of publishing any positive alternative bid is bounded by Lemma~\ref{lem:game:positive-entry-bound}. Since $u_{\dtr}^{k,\gamma}>u_{\mathrm{acc}}^{k,\gamma}$, no deviation can improve on deterrence.
\begin{showWhenSubmit}
    The proof is in the full version of this report~\cite{artifact}.
\end{showWhenSubmit}
\proofloc{app:proof:multistep}

The full game is itself the all-zero subgame beginning at step $t=1$, so the above theorem induces a deterrence equilibrium for the whole game.

\begin{corollary}[Multi-step deterrence SPNE]
\label{cor:game:multistep}
If
\[
u^{k,\gamma}_{\dtr}>u^{k,\gamma}_{\mathrm{acc}},
\]
then the full game $\mathcal G_m^{k,\gamma}$ has an SPNE in which
\[
b_1^{(j)}=b^{k,\gamma}_{\dtr}
\quad\text{and}\quad
b_2^{(j)}=0
\qquad
\text{for every }j=1,\ldots,m.
\]
\end{corollary}

In all, our game-theoretic analysis indicates that, whenever deterrence is more profitable than accommodation, the user would publish the deterring bid, and the attacker would abstain.
No rational participant gains by deviating.

  %%%%%%%%%%%%%%%%%%%%%%%%%%%%%%%%%%%%%%%%%%%%%%%%%%%%%%%%%%%%%%%%%%%%%%%%%%%%%%%%
  %%%%%%%%%%%%%%%%%%%%%%%%%%%%%%%%%%%%%%%%%%%%%%%%%%%%%%%%%%%%%%%%%%%%%%%%%%%%%%%%
  %%%%%%%%%%%%%%%%%%%%%%%%%%%%%%%%%%%%%%%%%%%%%%%%%%%%%%%%%%%%%%%%%%%%%%%%%%%%%%%%
      \section{Parameter Selection and User Revenue}\label{sec:tradeoff}

  Corollary~\ref{cor:game:explicit-k} only gives a sufficient condition on $k$ for deterrence to be the equilibrium at a specific losing-fee rate~$\gamma$.
  We first show how to choose a single $k$ that covers a range of $\gamma$ values~(\S\ref{sec:tradeoff:feasible}).
  We then show how the user's revenue grows with $\gamma$~(\S\ref{sec:tradeoff:design}), and derive the $k$ that maximizes the user's revenue~(\S\ref{sec:tradeoff:optimization}).

  %%%%%%%%%%%%%%%%%%%%%%%%%%%%%%%%%%%%%%%%%%%%%%%%%%%%%%%%%%%%%%%%%%%%%%%%%%%%%%%%
  %%%%%%%%%%%%%%%%%%%%%%%%%%%%%%%%%%%%%%%%%%%%%%%%%%%%%%%%%%%%%%%%%%%%%%%%%%%%%%%%
    \subsection{Choosing $k$ in Practice}\label{sec:tradeoff:feasible}

  For a fixed $\gamma$, our analysis gives two bounds on $k$.
  The Anti-Spam bound $k \ge \ln 2/\ln(1+\gamma)$ (Theorem~\ref{thm:sybil:tight}) is tight and hence necessary.
  For deterrence to be the equilibrium, the exact condition is $u^{k,\gamma}_{\dtr} > u^{k,\gamma}_{\mathrm{acc}}$ (Theorem~\ref{thm:2sgame:eq}); since $u^{k,\gamma}_{\mathrm{acc}}$ has no closed form, we instead use the sufficient bound $k \ge \max\{2,\exp(1/\gamma)/\gamma\}$ (Corollary~\ref{cor:game:explicit-k}).
  Clearing both is safe, and since both grow as $\gamma$ shrinks, choosing $k$ for the smallest $\gamma$ in a target range covers the whole range.

  We now turn to choosing $k$ in practice to deter front-running.
  We first measure the losing-fee rate $\gamma$ for arbitrage on Ethereum's Uniswap V2~\cite{UniswapV2} and V3~\cite{UniswapV3} (\S\ref{app:practical:gamma}).
  To defend against front-running on both, the designer must cover the smaller of the two rates, $\gamma = 0.259$.
  At this~$\gamma$, the two bounds we must clear are the Anti-Spam bound $k \ge \ln 2/\ln(1+\gamma) \approx 3.010$ and the sufficient deterrence bound $k \ge \max\{2,\exp(1/\gamma)/\gamma\} \approx 184$.
  Clearing both requires their maximum, $k \ge 184$, an impractically large exponent dictated entirely by the sufficient bound.

  Setting $k = 184$, Equation~\ref{eq:2sgame:b1-det} gives a deterrence bid $b^{k,\gamma}_{\dtr} \approx 0.974R$, so the user is left with almost nothing, a revenue of only $u^{k,\gamma}_{\dtr} = R - b^{k,\gamma}_{\dtr} \approx 0.026R$.
  This is no accident: By Equation~\ref{eq:2sgame:b1-det}, the deterrence bid approaches $R$ as $k$ grows, so an unnecessarily large $k$ drives the user's revenue toward zero.

  However, we show that this sufficient bound is far looser than necessary.
  We conduct a numerical simulation (\S\ref{app:empirical-k}) to find the smallest $k$ that meets our requirements.
  For ${\gamma = 0.259}$, the Anti-Spam bound requires only ${k \ge 3.010}$, which already satisfies the exact deterrence condition ${u^{k,\gamma}_{\dtr} > u^{k,\gamma}_{\mathrm{acc}}}$ for both Uniswap V2 and V3.

  We now compute the user's revenue at ${k=3.010}$.
  This revenue depends on the losing-fee rate $\gamma$: We compute the deterrence bid $b^{k,\gamma}_{\dtr}$ from Equation~\ref{eq:2sgame:b1-det} and the user's revenue is then given by the difference $u^{k,\gamma}_{\dtr}=R-b^{k,\gamma}_{\dtr}$.
  For a V2 arbitrage, $\gamma=0.259$ gives $b^{k,\gamma}_{\dtr}\approx0.830R$, so the user earns revenue $u^{k,\gamma}_{\dtr}=R-b^{k,\gamma}_{\dtr}\approx0.170R$; for a V3 arbitrage, $\gamma=0.867$ gives $b^{k,\gamma}_{\dtr}\approx0.555R$, so the user earns revenue $u^{k,\gamma}_{\dtr}=R-b^{k,\gamma}_{\dtr}\approx0.445R$.
  This is much better than the revenue of only $0.026R$ when selecting $k=184$.

  %%%%%%%%%%%%%%%%%%%%%%%%%%%%%%%%%%%%%%%%%%%%%%%%%%%%%%%%%%%%%%%%%%%%%%%%%%%%%%%%
  %%%%%%%%%%%%%%%%%%%%%%%%%%%%%%%%%%%%%%%%%%%%%%%%%%%%%%%%%%%%%%%%%%%%%%%%%%%%%%%%
  \subsection{Design Implications for $\gamma$}\label{sec:tradeoff:design}

  The example above already hints that a larger $\gamma$ raises the user's revenue: Holding $k=3.010$, the revenue grows from about $0.170R$ to $0.445R$ as $\gamma$ rises from $0.259$ to $0.867$.
  This holds in general: The revenue rises with $\gamma$ at every admissible $k$.

  \begin{proposition}[Revenue monotonicity at every $k$]\label{prop:tradeoff:monotone}
    Fix any $k > 1$. The deterrence bid $b^{k,\gamma}_{\dtr}(k, \gamma)$ is strictly decreasing in $\gamma > 0$, and consequently the user's revenue $u^{k,\gamma}_{\dtr}(k, \gamma) = R - b^{k,\gamma}_{\dtr}(k, \gamma)$ is strictly increasing in $\gamma$.
  \end{proposition}

  This follows directly from the closed form of the deterrence bid (Equation~\ref{eq:2sgame:b1-det}): $b^{k,\gamma}_{\dtr}\propto\gamma^{-1/k}$ is strictly decreasing in $\gamma$ for every $k>1$, so $u^{k,\gamma}_{\dtr}=R-b^{k,\gamma}_{\dtr}$ strictly increases.

  The implication for system design is direct: Although $\gamma$ is not freely tunable, charging losing transactions a higher~$\gamma$ raises user revenue at every~$k$.

  %%%%%%%%%%%%%%%%%%%%%%%%%%%%%%%%%%%%%%%%%%%%%%%%%%%%%%%%%%%%%%%%%%%%%%%%%%%%%%%%
  %%%%%%%%%%%%%%%%%%%%%%%%%%%%%%%%%%%%%%%%%%%%%%%%%%%%%%%%%%%%%%%%%%%%%%%%%%%%%%%%
  \subsection{Optimization of User Revenue}\label{sec:tradeoff:optimization}

  The previous analysis sets $k$ only to ensure Anti-Spam and to make deterrence the equilibrium.
  We now determine, for a fixed $\gamma$, the choice of $k$ that maximizes the user's deterrence revenue.

  \begin{proposition}[Revenue-optimal $k$]\label{prop:tradeoff:optimal-k}
    Fix $\gamma > 0$. The user's deterrence revenue $u^{k,\gamma}_{\dtr}(k,\gamma)$, viewed as a function of ${k>1}$, is strictly increasing on $\bigl(1,\,1+\frac{1}{\gamma}\bigr)$ and strictly decreasing on $\bigl(1+\frac{1}{\gamma},\infty\bigr)$. Its maximum, attained uniquely at $k = 1 + \frac{1}{\gamma}$, is
    \begin{equation*}
      u^{k,\gamma}_{\dtr}\Bigl(1+\tfrac{1}{\gamma},\,\gamma\Bigr) \;=\; \frac{\gamma}{1+\gamma}\,R .
    \end{equation*}
  \end{proposition}

  Intuitively, we track the trend through the deterrence bid: Its derivative in $k$ shows that $b^{k,\gamma}_{\dtr}$ first decreases and then increases. 
  Since $u^{k,\gamma}_{\dtr} = R - b^{k,\gamma}_{\dtr}$, the revenue follows the opposite trend, rising and then falling.
  \begin{showWhenSubmit}
    The proof is in the full version of this report~\cite{artifact}.
  \end{showWhenSubmit}
  \proofloc{app:proof:optimal-k}

  The revenue-optimal exponent $k = 1+\frac{1}{\gamma}$ depends on~$\gamma$, so no single $k$ maximizes revenue across a range of $\gamma$ values.
  A designer who chooses a single $k$ therefore trades off the revenues of users facing different $\gamma$ values.

  %%%%%%%%%%%%%%%%%%%%%%%%%%%%%%%%%%%%%%%%%%%%%%%%%%%%%%%%%%%%%%%%%%%%%%%%%%%%%%%%
  %%%%%%%%%%%%%%%%%%%%%%%%%%%%%%%%%%%%%%%%%%%%%%%%%%%%%%%%%%%%%%%%%%%%%%%%%%%%%%%%
  %%%%%%%%%%%%%%%%%%%%%%%%%%%%%%%%%%%%%%%%%%%%%%%%%%%%%%%%%%%%%%%%%%%%%%%%%%%%%%%%
    \section{Sandwich Resistance}\label{sec:sandwich}

  We finally show that PRECEDE also defends against the \emph{sandwich attack}~\cite{zhou2021high}, a prominent attack built on front-running.
  We model the attack (\S\ref{sec:sandwich:model}), show that under PRECEDE defending against a sandwich attack is no harder than defending against front-running (\S\ref{sec:sandwich:reduction}), and conclude that the user stops it at a cost below the loss she would otherwise suffer (\S\ref{sec:sandwich:comparison}).

  %%%%%%%%%%%%%%%%%%%%%%%%%%%%%%%%%%%%%%%%%%%%%%%%%%%%%%%%%%%%%%%%%%%%%%%%%%%%%%%%%%%%%%%%%
  %%%%%%%%%%%%%%%%%%%%%%%%%%%%%%%%%%%%%%%%%%%%%%%%%%%%%%%%%%%%%%%%%%%%%%%%%%%%%%%%%%%%%%%%%
  \subsection{Sandwich Model}\label{sec:sandwich:model}

  In a sandwich attack, the user publishes a transaction~$\tx_U$ that buys an asset (e.g., a swap on an automated decentralized exchange~\cite{angeris2020improved}), and an attacker brackets it with two transactions of her own.
  The front-running transaction $\tx_{F}$ buys the same asset ahead of $\tx_U$ and raises its price, so the user overpays; the back-running transaction $\tx_{B}$ then sells at this raised price, and the user's overpayment becomes the attacker's profit~\cite{zhou2021high,heimbach2022eliminating}.
  If the attack succeeds, this trading profit $V > 0$ is exactly the user's loss~\cite{heimbach2022eliminating}; here $V$ counts only the trade, neglecting transaction fees.

  Since PRECEDE orders at random, the attacker can rarely place $\tx_{F}$ and $\tx_{B}$ in the slots immediately adjacent to $\tx_U$.
  Following Tumas et al.~\cite{tumas2024ammazing}, we therefore grant the attacker the weakest success condition: The sandwich succeeds whenever $\tx_{F}$ is ordered before $\tx_U$ \emph{and} $\tx_{B}$ after it, regardless of any transactions ordered in between.
  
  %%%%%%%%%%%%%%%%%%%%%%%%%%%%%%%%%%%%%%%%%%%%%%%%%%%%%%%%%%%%%%%%%%%%%%%%%%%%%%%%%%%%%%%%%
  %%%%%%%%%%%%%%%%%%%%%%%%%%%%%%%%%%%%%%%%%%%%%%%%%%%%%%%%%%%%%%%%%%%%%%%%%%%%%%%%%%%%%%%%%
  \subsection{Reduction to Front-Running}\label{sec:sandwich:reduction}

  We now show that, under PRECEDE, defending against the sandwich attack is no harder than defending against front-running with reward $R = V$.

  First, the back-run is free.
  The back-run only needs~$\tx_{B}$ ordered after $\tx_U$, and $\tx_{B}$ needs no weight to achieve this: Since $\tx_U$ carries a positive bid, a bid close to~$0$ already places $\tx_{B}$ after $\tx_U$ with probability close to one.
  The back-run cost is thus negligible.

  What remains matches our front-running model.
  Being ordered first is worth $V$ to both sides: If $\tx_{F}$ precedes~$\tx_U$, the attacker collects $V$ and the user loses it; if $\tx_U$ stays first, the attacker gets nothing.
  This is precisely a single competition~(\S\ref{sec:model:execution}) with reward $R = V$, awarded to whichever of the user's $\tx_U$ and the attacker's $\tx_{F}$ is ordered first.
  The whole front-running analysis therefore applies with~$R$ replaced by~$V$.

  %%%%%%%%%%%%%%%%%%%%%%%%%%%%%%%%%%%%%%%%%%%%%%%%%%%%%%%%%%%%%%%%%%%%%%%%%%%%%%%%%%%%%%%%%
  %%%%%%%%%%%%%%%%%%%%%%%%%%%%%%%%%%%%%%%%%%%%%%%%%%%%%%%%%%%%%%%%%%%%%%%%%%%%%%%%%%%%%%%%%
  \subsection{Cheaper Defense by PRECEDE}\label{sec:sandwich:comparison}

  By the reduction, the user defends against sandwich attacks exactly as she defends against front-running: She publishes the deterrence bid $b^{k,\gamma}_{\dtr}$ (Equation~\ref{eq:2sgame:b1-det}) evaluated at ${R = V}$, and no attack is profitable.
  When ${k \ge \ln 2/\ln(1+\gamma)}$, this bid is strictly below~$V$ (Proposition~\ref{prop:det-profitable}), so the user defends at a cost smaller than the loss~$V$ she would otherwise suffer.
  Empirically, with $k=3.010$ covering Uniswap V2 and V3~(\S\ref{sec:tradeoff:feasible}), the user stops a sandwich at a cost of $0.830V$ on V2 and $0.555V$ on V3, well below the loss $V$ she would suffer in practical blockchains.

  Practical blockchains order by descending bid~\cite{buterin2013ethereum,yakovenko2018solana}, which is the $k \to \infty$ limit of PRECEDE~(\S\ref{sec:precede:intuition}).
  In this limit, the deterrence bid is $V$ itself.
  To keep $\tx_U$ ahead of the front-run, the user must then pay her entire potential loss: Any smaller bid lets the attacker outbid her and profit, so the user would lose both $V$ and her defense cost.
  The defense under PRECEDE is therefore strictly cheaper than under practical systems.

  % At the opposite extreme, uniform random ordering (${k \to 0}$), the attacker would instead spam many \hcAdd{front-running transactions to execute the sandwich attack on the user}~\cite{tumas2024ammazing}, which PRECEDE's Anti-Spam guarantee already rules out~(Theorem~\ref{thm:sybil:tight}).

  %%%%%%%%%%%%%%%%%%%%%%%%%%%%%%%%%%%%%%%%%%%%%%%%%%%%%%%%%%%%%%%%%%%%%%%%%%%%%%%%%%%%%%
  %%%%%%%%%%%%%%%%%%%%%%%%%%%%%%%%%%%%%%%%%%%%%%%%%%%%%%%%%%%%%%%%%%%%%%%%%%%%%%%%%%%%%%
  %%%%%%%%%%%%%%%%%%%%%%%%%%%%%%%%%%%%%%%%%%%%%%%%%%%%%%%%%%%%%%%%%%%%%%%%%%%%%%%%%%%%%%
    \section{Conclusion}\label{sec:conclusion}
  Front-running is a central security problem for blockchains.
  Despite theoretical work and practical mitigation efforts, existing defenses fall short, as enforcing a causal order is fundamentally hard.
  We presented PRECEDE, which prevents front-running by removing its incentive rather than trusting any party.
  By ordering transactions through a power-weighted randomized lottery, PRECEDE allows a user to deter every entrant with a single bid.
  PRECEDE guarantees Causal Ordering, Anti-Spam, Profitability, and in particular defense against sandwich attacks. 
  A user playing this deterrence strategy while the attacker abstains forms a subgame-perfect Nash equilibrium, so no rational attacker can profit by deviating.
  Every censorship-resistant blockchain can benefit from PRECEDE simply by updating its ordering function.
  
  % Our analysis also offers system designers a lever: Although $\gamma$ is not freely tunable, a larger $\gamma$ raises user revenue.

  \bibliographystyle{IEEEtran}
  \bibliography{reference}

\appendices

\section{Practical Losing-Fee Rates}\label{app:practical:gamma}

    \begin{table*}[t]
    \centering
    \footnotesize
    \setlength{\tabcolsep}{4pt}
    \begin{tabular}{lcccc}
    \toprule
    \textbf{Pool Name} & \textbf{Success Gas} & \textbf{Failure Gas} & $\boldsymbol{\gamma}$ & \textbf{Contract Address} \\
    \midrule
    \multicolumn{5}{l}{\textit{Uniswap V2}} \\
    USDT/WETH   & $[121{,}047,\ 121{,}059]$   & $[31{,}330,\ 31{,}342]$   & $[0.259,\ 0.259]$ & \texttt{0x0d4a11d5eeaac28ec3f61d100daf4d40471f1852} \\
    USDC/WETH  & $[120{,}603,\ 120{,}615]$   & $[31{,}360,\ 31{,}372]$   & $[0.260,\ 0.260]$ & \texttt{0xb4e16d0168e52d35cacd2c6185b44281ec28c9dc} \\
    \midrule
    \multicolumn{5}{l}{\textit{Uniswap V3}} \\
    USDT/WETH 0.01\%  & $[130{,}219,\ 131{,}545]$   & $[117{,}826,\ 118{,}247]$ & $[0.898,\ 0.905]$ & \texttt{0xc7bbec68d12a0d1830360f8ec58fa599ba1b0e9b} \\
    USDC/WETH 0.01\%  & $[129{,}565,\ 130{,}617]$   & $[116{,}214,\ 141{,}948]$ & $[0.890,\ 1.087]$ & \texttt{0xe0554a476a092703abdb3ef35c80e0d76d32939f} \\
    USDC/WETH 0.05\%  & $[120{,}911,\ 130{,}649]$   & $[113{,}466,\ 123{,}078]$ & $[0.890,\ 1.010]$ & \texttt{0x88e6a0c2ddd26feeb64f039a2c41296fcb3f5640} \\
    USDT/WETH 0.05\%  & $[121{,}912,\ 131{,}678]$   & $[115{,}436,\ 124{,}149]$ & $[0.877,\ 1.017]$ & \texttt{0x11b815efb8f581194ae79006d24e0d814b7697f6} \\
    WBTC/WETH 0.05\%  & $[113{,}934,\ 122{,}696]$   & $[105{,}555,\ 115{,}121]$ & $[0.867,\ 1.010]$ & \texttt{0x4585fe77225b41b697c938b018e2ac67ac5a20c0} \\
    USDT/WBTC 0.05\%  & $[125{,}941,\ 134{,}727]$   & $[117{,}524,\ 127{,}214]$ & $[0.873,\ 1.009]$ & \texttt{0x56534741cd8b152df6d48adf7ac51f75169a83b2} \\
    USDC/USDT 0.01\%  & $[133{,}600,\ 133{,}698]$   & $[126{,}125,\ 126{,}223]$ & $[0.944,\ 0.944]$ & \texttt{0x3416cf6c708da44db2624d63ea0aaef7113527c6} \\
    DAI/WETH 0.05\%   & $[120{,}054,\ 150{,}186]$   & $[105{,}877,\ 159{,}114]$ & $[0.881,\ 1.059]$ & \texttt{0x60594a405d53811d3bc4766596efd80fd545a270} \\
    \bottomrule
    \end{tabular}
    \vspace{2mm}
    \caption{Practical $\gamma$ on a mainnet fork (10 block samples per pool).}
    \label{tab:practical-gamma}
  \end{table*}
  We measure practical values of the losing-fee rate $\gamma$ for arbitrage on Ethereum.
  A transaction consumes resources measured in \emph{gas} and bids a price per unit of gas, so its total payment is the product of the gas price and the gas it consumes; this is the bid $b$ in our model.
  Winning the front-running competition, it executes fully and pays this $b$; losing, it reverts after consuming a ratio~$\gamma$ of that gas, and so pays $\gamma b$.
  Therefore, $\gamma$ is simply the ratio of the gas a transaction consumes if it loses to that if it wins.

  The value of $\gamma$ varies across arbitrages, but a participant can simulate her transaction before publishing it to find $\gamma$ in advance, so the user can set her deterrence bid accordingly.

  An arbitrage is simply a token swap, so we measure~$\gamma$ by running swaps on a local simulator (an Ethereum~\cite{buterin2013ethereum} mainnet fork served by Foundry's Anvil~\cite{foundry}), across the ten most-traded pools of Uniswap, the leading decentralized exchange on Ethereum, in~2025, of which two use version~2~(V2)~\cite{UniswapV2} and eight use version~3~(V3)~\cite{UniswapV3}, listed in Table~\ref{tab:practical-gamma}.
  Each pool swaps one pair of tokens at a fixed trading fee rate; since~V3 offers the same pair in several pools at different rates, we annotate each pool with its rate.
  We uniformly sample 10 blocks at random from all Ethereum blocks in 2025.
  For each pool, at each sampled block, we send five identical swaps: The first succeeds and the remaining four fail.
  Every pool pairs tokens among WETH and WBTC, the digital tokens standing for Ether~\cite{buterin2013ethereum} and Bitcoin~\cite{nakamoto2008bitcoin}, and US-dollar stablecoins (USDT, USDC, DAI), digital tokens each worth one US dollar.
  We publish each swap with a fixed input amount: 1~WETH or 1~WBTC when the pool holds ether or bitcoin, and 1{,}000~USDT in the stablecoin-only pool.

  In Uniswap~V3, a failing swap follows a different computation path than a successful one and can consume more gas.
  A rational participant caps her exposure with a gas limit, but cannot set it to the exact success cost: A successful execution needs some headroom to complete, so she sets the limit somewhat above the gas a success consumes, and a failing swap may consume up to this higher limit, pushing~$\gamma$ slightly above $1$.

  In every pool and block, the four failing swaps consume exactly the same amount of gas as one another.
  Across the two~V2 pools, $\gamma$ is stable at about $0.26$, that is, a failing swap uses less gas than a successful one.
  Across the eight~V3 pools, $\gamma$ is more variable, ranging over~$[0.867, 1.087]$, that is, a failing swap uses nearly as much gas as a successful one.
  Table~\ref{tab:practical-gamma} details, for each pool, the minimum and maximum of gas consumption and $\gamma$ over all 100 samples (10 blocks for each of the 10 pools).
  \begin{showWhenSubmit}
    The code is in our anonymous repo~\cite{artifact}.
\end{showWhenSubmit}

\iffullversion

\section{Proof of Lemma~\ref{lem:2sgame:deterrence}}\label{app:proof:2sgame-deterrence}

\begin{restated}{Lemma}{lem:2sgame:deterrence}{Deterrence weight}
    Fix the exponent parameter $k > 1$ and the losing-fee rate $\gamma > 0$, and let~$W$ be the total weight of transactions the attacker observes.
    Then the attacker's expected revenue with any single bid $b > 0$ satisfies $ u(b; W) \le 0 $ if and only if
    \begin{equation*}
        W \;\ge\; \frac{(k-1)^{k-1}}{\gamma\,k^{k}}\,R^{k}.
    \end{equation*}
\end{restated}

  \begin{proof}
    For $b > 0$, the factors $b$ and $b^{k} + W$ in Equation~\ref{eq:payoff:u} are positive, so $u(b; W)$ has the same sign as $R\,b^{k-1} - b^{k} - \gamma\,W$. Hence $u(b; W) \le 0$ for all $b > 0$ if and only if
    \begin{equation}\label{eq:det:proof-cond}
      R\,b^{k-1} - b^{k} \;\le\; \gamma\,W \qquad \text{for all } b > 0.
    \end{equation}
    
    That is, if and only if $\gamma\,W$ is at least the maximum of $g(b) \coloneqq R\,b^{k-1} - b^{k}$ over $b > 0$.
    We compute this maximum through differentiation.
    The derivative of $g$ is
    \begin{equation*}
      g'(b) = (k-1)R\,b^{k-2} - k\,b^{k-1} = b^{k-2}\bigl((k-1)R - k\,b\bigr) .
    \end{equation*}
    For $b > 0$ the factor $b^{k-2}$ is positive, so $g'(b)$ has the sign of $(k-1)R - k\,b$, which is positive for $b < \frac{k-1}{k}\,R$ and negative for $b > \frac{k-1}{k}\,R$. We denote this point by
    \begin{equation*}
      b^{\star} \;\coloneqq\; \frac{k-1}{k}\,R .
    \end{equation*}
    Thus $g(b)$ increases up to $b^{\star}$ and decreases afterwards, attaining its maximum at $b^{\star}$. Since $R - b^{\star} = R/k$, the maximum of $g(b)$ is
    \begin{equation*}
      \begin{aligned}
        \max_{b > 0} g(b) =& g(b^{\star})\\
        =& (b^{\star})^{k-1}\,(R - b^{\star})\\
        =& \left(\frac{k-1}{k}\,R\right)^{k-1}\frac{R}{k} \\
        =& \frac{(k-1)^{k-1}}{k^{k}}\,R^{k} .
      \end{aligned}
    \end{equation*}
    Substituting this into~Equation~\ref{eq:det:proof-cond}, the condition ${u(b; W) \le 0}$ for all $b > 0$ becomes $\gamma\,W \ge \frac{(k-1)^{k-1}}{k^{k}}\,R^{k}$.
    Therefore, $u(b; W) \le 0$ for all $b > 0$ if and only if
    \begin{equation*}
      W \;\ge\; \frac{(k-1)^{k-1}}{\gamma\,k^{k}}\,R^{k} . \qedhere
    \end{equation*}
  \end{proof}

\section{Proof of Theorem~\ref{thm:sybil:tight}}\label{app:proof:sybil-tight}

\begin{restated}{Theorem}{thm:sybil:tight}{Anti-Spam Threshold}
    Fix the exponent parameter $k>1$ and the losing-fee rate $\gamma>0$.
    For every bid vector $\bidvector_1$ of $P_1$ and every profile $\bidvector_{-1}$ of the other participants' bids, replacing $\bidvector_1$ with the single bid $(\bar b_1)$ does not decrease $P_1$'s revenue,
    \begin{equation*}
      u_1\bigl((\bar b_1), \bidvector_{-1}\bigr) \;\geq\; u_1(\bidvector_1, \bidvector_{-1}),
    \end{equation*}
    if and only if
    \begin{equation*}
      k \;\geq\; \frac{\ln 2}{\ln(1+\gamma)} .
    \end{equation*}
\end{restated}

\begin{proof}
% \begin{fullproof}
We show that the single bid $(\bar b_1)$ weakly dominates the original split vector $\bidvector_1$ for $P_1$.
We first reduce this revenue comparison to a single inequality between expected payments, and then treat the cases $\gamma \ge 1$ and~$0 < \gamma < 1$ separately.

The single bid and the split win with the same probability, since they have the same total weight:
\begin{equation*}
    \sum_{j=1}^{m_1} q(b_1^{(j)}, \bidvector_{-1}, \bidvector_1)=
    q(\bar b_1, \bidvector_{-1}, (\bar b_1)).
\end{equation*}

Recall that $P_1$'s revenue is her expected reward minus her expected payment $C_1$ (Equation~\ref{eq:sybil:utility}). Since the winning probability is unchanged, the difference in revenue between the single bid and the split is
\begin{align*}
    &u_1\bigl((\bar b_1),\bidvector_{-1}\bigr) - u_1(\bidvector_1,\bidvector_{-1})\\
    =& R\,q(\bar b_1, \bidvector_{-1}, (\bar b_1)) - R\sum_{j=1}^{m_1} q(b_1^{(j)}, \bidvector_{-1}, \bidvector_1) \\
    & \quad + C_1(\bidvector_1, \bidvector_{-1}) - C_1\bigl((\bar b_1), \bidvector_{-1}\bigr)\\
    =& C_1(\bidvector_1, \bidvector_{-1}) - C_1\bigl((\bar b_1), \bidvector_{-1}\bigr).
\end{align*}

We expand both payments by Equation~\ref{eq:sybil:Ci}.
Under the split $\bidvector_1$, Equation~\ref{eq:sybil:Ci} gives
\[
    \begin{aligned}
        C_1(\bidvector_1,\bidvector_{-1})
        &=
        \sum_{j=1}^{m_1}\Bigl(\gamma\, b_1^{(j)} + (1-\gamma)\, q(b_1^{(j)}, \bidvector_{-1}, \bidvector_1)\, b_1^{(j)}\Bigr) \\
        &=
        \gamma\sum_{j=1}^{m_1} b_1^{(j)}
        +
        (1-\gamma)\sum_{j=1}^{m_1} q(b_1^{(j)}, \bidvector_{-1}, \bidvector_1)\, b_1^{(j)} \\
        &=
        \gamma\sum_{j=1}^{m_1} b_1^{(j)}
        +
        (1-\gamma)\frac{\sum_{j=1}^{m_1}(b_1^{(j)})^{k+1}}{\sum_{i=1}^{n}\sum_{l=1}^{m_i}(b_i^{(l)})^k} .
    \end{aligned}
\]
Denote $P_1$'s total weight by $S \coloneqq \sum_{j=1}^{m_1}(b_1^{(j)})^k = \bar b_1^k$, and the total weight from other participants by $\Omega \coloneqq \sum_{i\neq 1}\sum_{l=1}^{m_i}(b_i^{(l)})^k$.
Denoting $L \coloneqq \sum_{j=1}^{m_1} b_1^{(j)}$ and ${\Lambda \coloneqq \sum_{j=1}^{m_1}(b_1^{(j)})^{k+1}}$, we obtain
\[
    C_1(\bidvector_1,\bidvector_{-1}) = \gamma L + (1-\gamma)\frac{\Lambda}{S+\Omega}.
\]

Under the surrogate $(\bar b_1)$, $P_1$ holds the single bid $\bar b_1$, so Equation~\ref{eq:sybil:Ci} gives
\[
    \begin{aligned}
        C_1\bigl((\bar b_1),\bidvector_{-1}\bigr)
        &=
        \gamma\bar b_1 + (1-\gamma)\, q\bigl(\bar b_1, \bidvector_{-1}, (\bar b_1)\bigr)\, \bar b_1 \\
        &=
        \gamma\bar b_1 + (1-\gamma)\frac{\bar b_1^k}{\bar b_1^k + \Omega}\, \bar b_1 \\
        &=
        \gamma\bar b_1 + (1-\gamma)\frac{S\bar b_1}{S+\Omega}.
    \end{aligned}
\]

Substituting the two payments into the revenue difference $C_1(\bidvector_1,\bidvector_{-1}) - C_1\bigl((\bar b_1),\bidvector_{-1}\bigr)$,
\[
    \begin{aligned}
        & u_1\bigl((\bar b_1),\bidvector_{-1}\bigr) - u_1(\bidvector_1,\bidvector_{-1}) \\
        =\;& C_1(\bidvector_1,\bidvector_{-1}) - C_1\bigl((\bar b_1),\bidvector_{-1}\bigr) \\
        =\;& \gamma(L - \bar b_1) - (1-\gamma)\frac{S\bar b_1 - \Lambda}{S+\Omega}.
    \end{aligned}
\]
Therefore, the single bid weakly dominates the split, $u_1\bigl((\bar b_1),\bidvector_{-1}\bigr) \ge u_1(\bidvector_1,\bidvector_{-1})$, if and only if this difference is non-negative, that is,
\begin{equation}\label{eq:sybil:tight-main}
    \gamma(L-\bar b_1)
    \ge
    (1-\gamma)\frac{S\bar b_1-\Lambda}{S+\Omega}.
\end{equation}

Since $b_1^{(j)}\le \bar b_1$ for every $j$, we have
\[
    \Lambda
    =
    \sum_{j=1}^{m_1}(b_1^{(j)})^k b_1^{(j)}
    \le
    \bar b_1\sum_{j=1}^{m_1}(b_1^{(j)})^k
    =
    S\bar b_1.
\]
Therefore, we have 
\begin{equation}\label{eq:sybil:tight-main-weak}
      S\bar b_1 - \Lambda \ge 0.
\end{equation}

Now we discuss the two cases for $\gamma$.
First, if $\gamma \ge 1$, the single bid weakly dominates the split for every $\bidvector_1$ and every profile $\bidvector_{-1}$.
The left-hand side $\gamma(L-\bar b_1)$ of Equation~\ref{eq:sybil:tight-main} is non-negative: Since $k > 1$ and $b_1^{(j)} \le L$ for every $j$,
\[
    \bar b_1^k = \sum_{j=1}^{m_1}(b_1^{(j)})^k \leq (\sum_{j=1}^{m_1}(b_1^{(j)}))^k = L^k ,
\]
so $\bar b_1 \le L$.
The right-hand side $(1-\gamma)(S\bar b_1-\Lambda)/(S+\Omega)$ is non-positive, since $1-\gamma \le 0$ and $S\bar b_1 - \Lambda \ge 0$ (Equation~\ref{eq:sybil:tight-main-weak}).
Hence Equation~\ref{eq:sybil:tight-main} holds.
This is consistent with Theorem~\ref{thm:sybil:tight}: For $\gamma \ge 1$ the Anti-Spam bound is automatically met, as $\ln 2/\ln(1+\gamma) \le 1 < k$.

Second, we consider the case where ${0 < \gamma < 1}$.
We first show that, for a fixed $k$, the single bid weakly dominates every split if and only if $\gamma$ is not below a threshold $\gamma_0$ that depends on $k$. 
We then deduce the stated bound on $k$.
We proceed in four steps.
Step~1 reformulates the dominance condition as a threshold on $\gamma$. 
Step~2 evaluates this supremum at the equal two-way split, which gives the necessary threshold $\gamma_0 = 2^{1/k}-1$.
Step~3 shows, by a merging argument, that no finer split requires a larger $\gamma$, so~$\gamma_0$ is also sufficient.
Step~4 returns to $k$, rearranging~$\gamma \ge \gamma_0$ into the bound $k \ge \ln 2/\ln(1+\gamma)$.

\emph{Step 1: Reformulation as a threshold on $\gamma$.}
Since ${1-\gamma > 0}$ and $S\bar b_1 - \Lambda \ge 0$ (Equation~\ref{eq:sybil:tight-main-weak}), the right-hand side of Equation~\ref{eq:sybil:tight-main} is non-negative and decreasing in $\Omega$. Hence, it holds for every $\Omega \ge 0$ if and only if it holds at $\Omega = 0$, where the right-hand side is largest and the inequality is hardest to satisfy.
It therefore suffices to prove the inequality at $\Omega = 0$,
\begin{equation}\label{eq:sybil:tight-main-0}
    \gamma(L-\bar b_1)
    \ge
    (1-\gamma)\Bigl(\bar b_1 - \frac{\Lambda}{S}\Bigr).
\end{equation}

We now reformulate the inequality in terms of $P_1$'s weight shares, the fractions of her total weight $S$ carried by her individual bids.
For each bid $j$, denote this share by~$p^{(j)}$, and let~$a \coloneqq 1/k$ be the reciprocal exponent:
\[
    a\coloneqq \frac{1}{k}\in(0,1),
    \qquad
    p^{(j)}
    \coloneqq
    \frac{(b_1^{(j)})^k}{S}
    =
    \left(\frac{b_1^{(j)}}{\bar b_1}\right)^k.
\]
We have
$$
    \frac{b_1^{(j)}}{\bar b_1}=(p^{(j)})^a,
$$
and
$$
    \sum_{j=1}^{m_1} p^{(j)}= \sum_{j=1}^{m_1} \frac{(b_1^{(j)})^k}{S}=1.
$$

Dividing Equation~\ref{eq:sybil:tight-main-0} by $\bar b_1$ gives
\[
    \gamma\Bigl(\frac{L}{\bar b_1} - 1\Bigr)
    \ge
    (1-\gamma)\Bigl(1 - \frac{\Lambda}{S\bar b_1}\Bigr).
\]
By the definitions of $L$ and $\Lambda$, with $S = \bar b_1^k$, ${p^{(j)} = (b_1^{(j)})^k/\bar b_1^k}$, and $b_1^{(j)}/\bar b_1 = (p^{(j)})^a$,
\[
    \begin{aligned}
        \frac{L}{\bar b_1}
        &= \sum_{j=1}^{m_1}\frac{b_1^{(j)}}{\bar b_1}
        = \sum_{j=1}^{m_1}(p^{(j)})^a, \\
        \frac{\Lambda}{S\bar b_1}
        &= \sum_{j=1}^{m_1}\frac{(b_1^{(j)})^k}{\bar b_1^k}\,\frac{b_1^{(j)}}{\bar b_1} \\
        &= \sum_{j=1}^{m_1} p^{(j)}\,(p^{(j)})^a
        = \sum_{j=1}^{m_1}(p^{(j)})^{1+a}.
    \end{aligned}
\]
Substituting these into the inequality above gives
\begin{equation}\label{eq:sybil:tight-normalized}
    \gamma\left(\sum_{j=1}^{m_1} (p^{(j)})^a-1\right)
    \ge
    (1-\gamma)\left(1-\sum_{j=1}^{m_1} (p^{(j)})^{1+a}\right).
\end{equation}

Let $p=(p^{(1)}, \ldots, p^{(m_1)})$ be the resulting share vector.
For a split, define
$
    {F(p)\coloneqq \sum_{j=1}^{m_1} (p^{(j)})^a-1},
$
and
$
    {G(p)\coloneqq 1-\sum_{j=1}^{m_1} (p^{(j)})^{1+a}}.
$
By the definitions of $F$ and~$G$, Equation~\ref{eq:sybil:tight-normalized} is exactly
$
    \gamma F(p) \ge (1-\gamma)G(p).
$
Both terms are positive: $F(p)>0$ since $x^a$ is strictly concave on $(0,1)$, and $G(p)>0$ since $(p^{(j)})^{1+a}<p^{(j)}$ for every component strictly between $0$ and $1$.
Hence $F(p)+G(p)>0$, so dividing by it, the inequality is equivalent to
\begin{equation}\label{eq:sybil:gamma-ratio}
    \gamma
    \ge
    \frac{G(p)}{F(p)+G(p)}.
\end{equation}
The single bid therefore weakly dominates \emph{every} split if and only if $\gamma$ satisfies Equation~\ref{eq:sybil:gamma-ratio} for all split vectors $p$, that is,
\[
    \gamma \ge \sup_{p}\frac{G(p)}{F(p)+G(p)}.
\]
This supremum is the exact threshold on $\gamma$.

\emph{Step 2: A necessary bound from the worst binary split.}
The threshold is the supremum over all splits, so any particular split lower-bounds it and thus yields a necessary condition on $\gamma$. 
We therefore probe with the simplest family, the two-bid splits $p=(z,1-z)$ for $z\in(0,1)$ (recall that $\sum_{j} p^{(j)} = 1$). 
Among these we identify the worst one, the split that maximizes the ratio in Equation~\ref{eq:sybil:gamma-ratio} and so requires the largest $\gamma$.
Specializing $F$ and $G$ to this two-component vector, denote
\begin{equation}\label{eq:sybil:F2G2}
    \begin{aligned}
        F_2(z) &\coloneqq F\bigl((z,1-z)\bigr) = z^a+(1-z)^a-1, \\
        G_2(z) &\coloneqq G\bigl((z,1-z)\bigr) = 1-z^{1+a}-(1-z)^{1+a}.
    \end{aligned}
\end{equation}
Both are symmetric under $z\mapsto 1-z$, so $(z,1-z)$ and its mirror image $(1-z,z)$ give the same ratio; we may therefore restrict to $z\in(0,1/2]$.
For this binary split, the required lower bound on $\gamma$ is
\[
    \frac{G_2(z)}{F_2(z)+G_2(z)}
    =
    \frac{G_2(z)/F_2(z)}{1+G_2(z)/F_2(z)}.
\]
This increases with the ratio $G_2(z)/F_2(z)$, so maximizing the bound is the same as maximizing $G_2(z)/F_2(z)$.

We show that $G_2(z)/F_2(z)$ is increasing on $(0,1/2]$. 
Both functions vanish at the left endpoint, since $F_2(0)=0^a+1-1=0$ and $G_2(0)=1-0-1=0$, so by the monotone form of l'H\^opital's rule it suffices to show that the ratio of derivatives $G_2'(z)/F_2'(z)$ is increasing on $(0,1/2]$.

Indeed,
\[
    F_2'(z)
    =
    a\left(z^{a-1}-(1-z)^{a-1}\right)>0,
\]
and
\[
    G_2'(z)
    =
    (1+a)\left((1-z)^a-z^a\right)>0.
\]
Dividing the two derivatives,
\[
    \frac{G_2'(z)}{F_2'(z)}
    =
    \frac{1+a}{a}\,
    \frac{(1-z)^a-z^a}{z^{a-1}-(1-z)^{a-1}}.
\]
Dividing the numerator and denominator by $(1-z)^a$,
\[
    \frac{G_2'(z)}{F_2'(z)}
    =
    \frac{1+a}{a}\,
    \frac{(1-z)\bigl[1-\bigl(\tfrac{z}{1-z}\bigr)^a\bigr]}{\bigl(\tfrac{z}{1-z}\bigr)^{a-1}-1}.
\]
Using $1-z=1/\bigl(1+\tfrac{z}{1-z}\bigr)$, the right-hand side depends on $z$ only through $\tfrac{z}{1-z}$,
\[
    \frac{G_2'(z)}{F_2'(z)}
    =
    \frac{1+a}{a}\,
    \frac{1-\bigl(\tfrac{z}{1-z}\bigr)^a}{\bigl(1+\tfrac{z}{1-z}\bigr)\bigl[\bigl(\tfrac{z}{1-z}\bigr)^{a-1}-1\bigr]}.
\]
Let $r=\tfrac{z}{1-z}$; then we have
\begin{equation}\label{eq:sybil:r}
    \frac{G_2'(z)}{F_2'(z)}
    =
    \frac{1+a}{a}\,
    \frac{1-r^a}{(1+r)\bigl(r^{a-1}-1\bigr)}.
\end{equation}

We now show the right-hand side of Equation~\ref{eq:sybil:r} is increasing in ${r\in(0,1]}$. 
Let $h(r)\coloneqq\dfrac{1-r^a}{(1+r)\bigl(r^{a-1}-1\bigr)}$. Differentiating,
\[
    h'(r)=\frac{1-r^{2a-2}+(1-a)\,r^{a-2}\bigl(1-r^2\bigr)}{(1+r)^2\bigl(r^{a-1}-1\bigr)^2} .
\]
The denominator is positive, so $h'$ has the sign of its numerator, which we denote $\psi(r)\coloneqq 1-r^{2a-2}+(1-a)\,r^{a-2}\bigl(1-r^2\bigr)$. 
It suffices to show $\psi>0$ on $(0,1)$: Then $h'>0$ on $(0,1)$, so $h$ is increasing on $(0,1]$.
We have $\psi(1)=0$ and
\[
    \psi'(r)=(1-a)\,r^{a-3}\bigl(2r^a-a r^2+(a-2)\bigr) .
\]
With $\chi(r)\coloneqq 2r^a-a r^2+(a-2)$ and $(1-a)\,r^{a-3}>0$ on~$(0,1)$, we have $\psi'(r)>0$ if and only if $\chi(r)>0$.
We have $\chi(1)=2-a+a-2=0$, and
\[
    \chi'(r)=2a\,r^{a-1}-2a\,r=2a\,r\,\bigl(r^{a-2}-1\bigr) .
\]
For $r\in(0,1)$ the exponent $a-2$ is negative, so ${r^{a-2}>1}$ and hence $\chi'(r)>0$. Thus $\chi$ is increasing on $(0,1]$, and since $\chi(1)=0$, we get $\chi<0$ on $(0,1)$.
Returning to $\psi'(r)=(1-a)\,r^{a-3}\,\chi(r)$, on $(0,1)$ the factors $1-a$ and~$r^{a-3}$ are positive while $\chi(r)<0$, so $\psi'(r)<0$. 
Thus~$\psi$ is decreasing on $(0,1]$, and since $\psi(1)=0$, we get $\psi>0$ on $(0,1)$, the inequality we needed.

It remains to transfer this back to $z$. Since $r=z/(1-z)$ is increasing in $z$ and maps $(0,1/2]$ onto $(0,1]$, $G_2'(z)/F_2'(z)$ is increasing on $(0,1/2]$. 
The rule above then lifts this to~$G_2(z)/F_2(z)$, which is increasing on $(0,1/2]$ and hence maximal at $z=1/2$. 
The binary ratio increases with $G_2/F_2$, so it too is maximal at $z=1/2$: The worst binary split is the equal split $z=1/2$.
At this split,
\[
    F_2(1/2)
    =
    2^{1-a}-1,
    \qquad
    G_2(1/2)
    =
    1-2^{-a},
\]
and hence the binary ratio at $z=1/2$ is
\begin{equation*}
    \frac{G_2(1/2)}{F_2(1/2)+G_2(1/2)}
    =
    2^a-1.
\end{equation*}
Denote $\gamma_0 \coloneqq 2^a-1$. 
Since the binary ratio is maximized at $z=1/2$, every $z\in(0,1)$ satisfies ${G_2(z)/(F_2(z)+G_2(z)) \le \gamma_0}$, that is,
\begin{equation}\label{eq:sybil:binary-bound}
    \gamma_0 F_2(z) \ge (1-\gamma_0)\,G_2(z).
\end{equation}
The worst binary split therefore requires $\gamma \ge \gamma_0$, the necessary bound on the threshold.

\emph{Step 3: Sufficiency by a merging argument.}
We now show that this binary-split bound is generally sufficient: For every split $p$, $\gamma_0$ is at least $\frac{G(p)}{F(p)+G(p)}$.
We turn to the equivalent form of this: Every split satisfies Equation~\ref{eq:sybil:tight-normalized} when $\gamma = \gamma_0$.

Setting $\gamma = \gamma_0$ in Equation~\ref{eq:sybil:tight-normalized} and using ${\gamma_0 + (1-\gamma_0) = 1}$, the inequality is equivalent to a sum over the shares,
\begin{equation}\label{eq:sybil:tight-normalized-2}
    \sum_{j=1}^{m_1}\bigl[\gamma_0 (p^{(j)})^a + (1-\gamma_0)(p^{(j)})^{1+a}\bigr] \ge 1.
\end{equation}
Denote the contribution of a single share by
\[
    f(x) \coloneqq \gamma_0 x^a + (1-\gamma_0)x^{1+a}.
\]
Note that $f(1) = 1$, so Equation~\ref{eq:sybil:tight-normalized-2} is exactly
\begin{equation}\label{eq:sybil:f-sum}
    \sum_{j=1}^{m_1} f(p^{(j)}) \ge f(1).
\end{equation}
We prove Equation~\ref{eq:sybil:f-sum} by a merging argument. 
Take any two components $x,y>0$, let $t=x+y\le 1$, and denote~$\zeta=x/t$, so that $x=\zeta t$ and $y=(1-\zeta)t$; then we have
\begin{align*}
    & f(x)+f(y)-f(t) \\
    =\;& \gamma_0 x^a + (1-\gamma_0)x^{1+a}+ \gamma_0 y^a + (1-\gamma_0)y^{1+a} \\
    & \quad {}- \gamma_0 t^a - (1-\gamma_0)t^{1+a} \\
    =\;& \gamma_0\bigl(x^a+y^a-t^a\bigr) \\
    & \quad + (1-\gamma_0)\bigl(x^{1+a}+y^{1+a}-t^{1+a}\bigr) \\
    =\;& \gamma_0 t^a\bigl(\zeta^a+(1-\zeta)^a-1\bigr) \\
    & \quad + (1-\gamma_0)t^{1+a}\bigl(\zeta^{1+a}+(1-\zeta)^{1+a}-1\bigr) \\
    \overset{\eqref{eq:sybil:F2G2}}{=}\;& \gamma_0 t^a F_2(\zeta) - (1-\gamma_0)t^{1+a}G_2(\zeta) \\
    =\;& t^a\bigl[\gamma_0 F_2(\zeta) - (1-\gamma_0)\,t\,G_2(\zeta)\bigr] \\
    \ge\;& t^a\bigl[\gamma_0 F_2(\zeta) - (1-\gamma_0)G_2(\zeta)\bigr] \\
    \overset{\eqref{eq:sybil:binary-bound}}{\ge}\;& 0.
\end{align*}
Therefore,
$
    f(x)+f(y)\ge f(x+y).
$
Iteratively merging the components of any split vector $p$ yields
\[
    \sum_{j=1}^{m_1} f(p^{(j)})
    \ge
    f\left(\sum_{j=1}^{m_1} p^{(j)}\right)
    =
    f(1)
    =
    1,
\]
which proves Equation~\ref{eq:sybil:f-sum}. 
Hence $\gamma\ge\gamma_0$ is sufficient: The single bid weakly dominates every split.

\emph{Step 4: Bound on $k$.}
Together with the necessary bound from the worst binary split, this shows that for $0{<\gamma<1}$ the single bid weakly dominates every split if and only if
\begin{equation*}
    \gamma\ge\gamma_0 = 2^{1/k}-1,
\end{equation*}
or, equivalently,
\begin{equation}\label{eq:sybil:gamma0-condition}
    k\ge \frac{\ln 2}{\ln(1+\gamma)} .
\end{equation}

\emph{Conclusion.}
For any $\gamma>0$, the single bid weakly dominates every split if and only if $k\ge \ln 2/\ln(1+\gamma)$: For $\gamma\ge1$ this holds automatically, as $\ln 2/\ln(1+\gamma)\le 1<k$, and for $0<\gamma<1$ it is the condition just derived (Equation~\ref{eq:sybil:gamma0-condition}).
% \end{fullproof}
\end{proof}

\section{Proof of Proposition~\ref{prop:det-profitable}}\label{app:proof:det-profitable}

\begin{restated}{Proposition}{prop:det-profitable}{Deterrence is profitable}
If
$
    {k \ge \frac{\ln 2}{\ln(1+\gamma)}},
$
then the deterrence strategy results in strictly positive revenue, namely
$
    u_{\dtr}^{k,\gamma}>0.
$
\end{restated}

\begin{proof}
% \begin{fullproof}
We show that Profitability holds exactly when the losing-fee rate satisfies $\gamma > (k-1)^{k-1}/k^k$, and that the Anti-Spam threshold is stricter, hence guarantees it.

Recall that the deterrence bid is
$
    {b_{\dtr}^{k,\gamma}
    =
    R\left(
        \frac{(k-1)^{k-1}}{\gamma k^k}
    \right)^{1/k}}
$ (Equation~\ref{eq:2sgame:b1-det}),
and hence the user's revenue under deterrence is
$
    u_{\dtr}^{k,\gamma}
    =
    R-b_{\dtr}^{k,\gamma}.
$
Therefore, deterrence is profitable if and only if
$
    b_{\dtr}^{k,\gamma}<R.
$
Since $R>0$, this condition is equivalent to
$
    \left(
        \frac{(k-1)^{k-1}}{\gamma k^k}
    \right)^{1/k}
    <1,
$
or, equivalently,
\begin{equation}\label{eq:det:profit-threshold}
    \gamma
    >
    \frac{(k-1)^{k-1}}{k^k}.
\end{equation}
Therefore, it is now sufficient to show that the above equation is satisfied when $k \geq \frac{\ln 2}{\ln(1+\gamma)}$, which is equivalent to~$\gamma \geq 2^{1/k}-1$.
Hence it suffices to prove that, for every~$k>1$,
\begin{equation}\label{eq:det:sybil-implies-profit}
    2^{1/k}-1
    >
    \frac{(k-1)^{k-1}}{k^k}.
\end{equation}
This inequality says exactly that the weakest losing-fee rate required for Anti-Spam is already strictly larger than the weakest losing-fee rate required for the deterrence bid to be below $R$.

Denote
$
    y\coloneqq \frac{1}{k}\in(0,1).
$
Then
\[
    \frac{(k-1)^{k-1}}{k^k}
    =
    y(1-y)^{(1-y)/y}.
\]
Thus Equation~\ref{eq:det:sybil-implies-profit} becomes
\begin{equation}\label{eq:det:y-ineq}
    2^y-1
    >
    y(1-y)^{(1-y)/y}.
\end{equation}

We first upper bound the right-hand side. Since ${\ln(1-y)<-y}$ for $y\neq 0$, multiplying by $(1-y)/y$, which is positive,
\[
    \begin{aligned}
        \frac{1-y}{y}\ln(1-y)
        &< \frac{1-y}{y}\,(-y) \\
        &= -(1-y).
    \end{aligned}
\]
Therefore,
\begin{equation}\label{eq:det:exp-upper}
    \begin{aligned}
        y(1-y)^{(1-y)/y}
        &=
        y\exp\left(\frac{1-y}{y}\ln(1-y)\right) \\
        &<
        y e^{-(1-y)}\\
        &=
        y e^{y-1}.
    \end{aligned}
\end{equation}

Therefore, to prove Equation~\ref{eq:det:y-ineq} it suffices to show that $y e^{y-1} < 2^y-1$.

Define $\xi(y)
    \coloneqq
    2^y-1-y e^{y-1}$.
Then
$
    \xi(0)=\xi(1)=0.
$
Moreover,
\[
    \xi''(y)
    =
    (\ln 2)^2 2^y-(y+2)e^{y-1}.
\]
We claim that $\xi''(y)<0$ for every $y\in[0,1]$. Since ${(y+2)e^{y-1}>0}$, $\xi''(y)<0$ is equivalent to
\[
    \frac{(\ln 2)^2 2^y}{(y+2)e^{y-1}} <1 .
\]
The logarithm of this ratio is $$2\ln(\ln 2)+y\ln 2-\ln(y+2)-(y-1).$$ Its derivative, the logarithmic derivative of the ratio, equals $ \ln 2-1-\frac{1}{y+2}<0$,
so this ratio is decreasing on~$[0,1]$, hence maximized at $y=0$. At $y=0$ it is
$
    {\frac{e(\ln 2)^2}{2}<1}.
$
Hence the ratio is strictly smaller than $1$ on $[0,1]$, which proves $\xi''(y)<0$.

Thus $\xi$ is strictly concave on $[0,1]$. Since ${\xi(0)=\xi(1)=0}$, strict concavity gives
$
    \xi(y)>0
$
for every $y\in(0,1)$.
Equivalently,
\begin{equation}\label{eq:det:middle-upper}
    y e^{y-1}
    <
    2^y-1
    \qquad\text{for all } y\in(0,1).
\end{equation}
Combining Equation~\ref{eq:det:exp-upper} and Equation~\ref{eq:det:middle-upper}, we obtain
\[
    y(1-y)^{(1-y)/y}
    <
    2^y-1,
\]
which proves Equation~\ref{eq:det:y-ineq}, and therefore Equation~\ref{eq:det:sybil-implies-profit}.

Consequently,
\[
    \gamma
    \ge
    2^{1/k}-1
    >
    \frac{(k-1)^{k-1}}{k^k}.
\]
By Equation~\ref{eq:det:profit-threshold}, this implies
$
    b_{\dtr}^{k,\gamma}<R.
$
Hence
$
    {u_{\dtr}^{k,\gamma}
    =
    R-b_{\dtr}^{k,\gamma}
    >
    0}.
$
This proves Profitability.
% \end{fullproof}
\end{proof}

\section{Proof of Lemma~\ref{lem:game:acc-bound}}\label{app:proof:acc-bound}

\begin{restated}{Lemma}{lem:game:acc-bound}{Accommodation utility bound}
    For any $\gamma > 0$ and $k>1$, the accommodation utility $u^{k,\gamma}_{\mathrm{acc}}$ satisfies
    \begin{equation*}
      u^{k,\gamma}_{\mathrm{acc}} \;\le\; \frac{R}{k\,\min\{1,\gamma\}}.
    \end{equation*}
\end{restated}

  \begin{proof}
% \begin{fullproof}
    We bound $P_1$'s utility by $q_1^* R$, and bound the winning probability $q_1^*$ using $P_2$'s first-order optimality condition.

    Fix any $b_1\in(0,b^{k,\gamma}_{\dtr})$, and recall that $b_2^*(b_1)$ is $P_2$'s selected best response~(\S\ref{sec:game:2seq:p2}).
    Since $b_1$ is below the deterring bid, Lemma~\ref{lem:2sgame:deterrence} implies that some positive bid gives $P_2$ positive utility.
    Therefore the best response also gives positive utility:
    \[
      u_2^*\coloneqq u_2(b_1,b_2^*(b_1))>0.
    \]
    In particular, $b_2^*(b_1)>0$.
    Moreover, no bid $b_2\ge R$ can be a best response.
    For such a bid, Equation~\ref{eq:2sgame:utility} gives $u_2(b_1,b_2)=q_2(R-b_2)-(1-q_2)\gamma b_2$, in which the winning payoff $R-b_2\le 0$ and the losing payoff $-\gamma b_2<0$.
    Because $b_1>0$ makes $q_1>0$, and hence $q_2<1$, the losing event carries positive weight $1-q_2>0$, so $u_2(b_1,b_2)<0$.
    Since $u_2^*>0$, the best response therefore satisfies
    \begin{equation}\label{eq:game:b2-below-R}
      b_2^*(b_1)<R.
    \end{equation}

    Let $q_i^*\coloneqq q_i(b_1,b_2^*(b_1))$
    be the winning probabilities at this best response.
    Since ${b_1>0}$ and $b_2^*(b_1)>0$, both probabilities are positive and $q_1^*+q_2^*=1$.

    We now express $P_2$'s utility~(Equation~\ref{eq:2sgame:utility}) as a function of her own bid, holding $b_1$ fixed:
    \begin{equation}\label{eq:game:u2-obj}
      u_2(b_1,b_2)
      =
      q_2(b_1,b_2)R
      -
      b_2\bigl[\gamma+(1-\gamma)q_2(b_1,b_2)\bigr].
    \end{equation}
    For $b_2>0$, since $q_2 = b_2^k/(b_1^k+b_2^k)$ (Equation~\ref{eq:2sgame:winprob}),
    \[
      \frac{\partial q_2}{\partial b_2}
      =
      \frac{k\,b_1^k\,b_2^{k-1}}{(b_1^k+b_2^k)^2}
      =
      \frac{k}{b_2}q_1q_2.
    \]
    Differentiating $u_2$~(Equation~\ref{eq:game:u2-obj}) term by term, with the product rule on~$b_2 q_2$,
    \begin{align*}
      \frac{\partial u_2}{\partial b_2}
      &=
      R\,\frac{\partial q_2}{\partial b_2} - \gamma - (1-\gamma)\Bigl(q_2 + b_2\,\frac{\partial q_2}{\partial b_2}\Bigr) \\
      &=
      \frac{\partial q_2}{\partial b_2}\bigl[R-(1-\gamma)b_2\bigr] - \gamma - (1-\gamma)q_2 \\
      &=
      \frac{k}{b_2}q_1q_2\bigl[R-(1-\gamma)b_2\bigr] - \gamma - (1-\gamma)q_2 .
    \end{align*}
    
    Since $b_2^*(b_1)>0$ lies in the interior of $[0,\infty)$ and maximizes~$u_2(b_1,\cdot)$, $\partial u_2/\partial b_2 = 0$; multiplying by $b_2^*(b_1)$ gives the first-order condition
    \begin{equation}\label{eq:game:u2-br-foc}
      kq_1^*q_2^*\bigl[R-(1-\gamma)b_2^*(b_1)\bigr]
      =
      b_2^*(b_1)\bigl[\gamma+(1-\gamma)q_2^*\bigr].
    \end{equation}

    Next substitute~Equation~\ref{eq:game:u2-br-foc} into $P_2$'s utility~(Equation~\ref{eq:game:u2-obj}):
    \begin{align}
      u_2^*
      &=
      q_2^*R
      -
      b_2^*(b_1)\bigl[\gamma+(1-\gamma)q_2^*\bigr] \nonumber\\
      &=
      q_2^*
      \Bigl\{
        R-kq_1^*\bigl[R-(1-\gamma)b_2^*(b_1)\bigr]
      \Bigr\}.
      \label{eq:game:u2-foc}
    \end{align}
    $u_2^*$ is the product of $q_2^*$ and ${R-kq_1^*\bigl[R-(1-\gamma)b_2^*(b_1)\bigr]}$ (Equation~\ref{eq:game:u2-foc}).
    Since this product is positive ($u_2^*>0$) and its first factor is positive ($q_2^*>0$), the second factor must be positive as well, that is,
    \[
      kq_1^*\bigl[R-(1-\gamma)b_2^*(b_1)\bigr]<R.
    \]
    Since $b_2^*(b_1)<R$ (Equation~\ref{eq:game:b2-below-R}), we bound $R-(1-\gamma)b_2^*(b_1)$ in two cases: For $0<\gamma\le 1$, $(1-\gamma)\bigl(R-b_2^*(b_1)\bigr)\ge 0$ gives $R-(1-\gamma)b_2^*(b_1)\ge\gamma R$; for $\gamma>1$, $R-(1-\gamma)b_2^*(b_1)=R+(\gamma-1)b_2^*(b_1)\ge R$ since $b_2^*(b_1)\ge0$. 
    In short,
    \[
      R-(1-\gamma)b_2^*(b_1)\ge\min\{1,\gamma\} R.
    \]
    Combining the last two inequalities yields
    \[
      q_1^*
      <
      \frac{R}{k\bigl[R-(1-\gamma)b_2^*(b_1)\bigr]}
      \le
      \frac{1}{k\,\min\{1,\gamma\}}.
    \]

    Finally, bound $P_1$'s utility at this accommodated outcome:
    \[
      u_1(b_1,b_2^*(b_1))
      =
      q_1^*(R-b_1)-q_2^*\gamma b_1
      \le
      q_1^*R
      <
      \frac{R}{k\,\min\{1,\gamma\}}.
    \]
    This holds for every $b_1\in(0,b^{k,\gamma}_{\dtr})$.
    Taking the supremum over exactly this interval gives
    $u^{k,\gamma}_{\mathrm{acc}}\le R/(k\,\min\{1,\gamma\})$.
% \end{fullproof}
  \end{proof}

\section{Proof of Lemma~\ref{lem:game:explicit-k-payoff}}\label{app:proof:explicit-k-payoff}

\begin{restated}{Lemma}{lem:game:explicit-k-payoff}{Deterrence utility bound}
    Fix $\gamma > 0$. If
    \[
      k \;\geq\; \max\left\{2,\frac{\exp(1/\gamma)}{\gamma}\right\},
    \]
    then the utility of the deterrence strategy exceeds $\frac{R}{k\,\min\{1,\gamma\}}$:
    \begin{equation}\label{eq:game:sufficient}
      u^{k,\gamma}_{\dtr} \;>\; \frac{R}{k\,\min\{1,\gamma\}}.
    \end{equation}
\end{restated}

  \begin{proof}
    We treat the two cases $\gamma>1$ and $0<\gamma\leq1$ separately.

    First suppose $\gamma>1$.
    The idea is to bound the deterrence bid by its value at $\gamma=1$, where it is largest, and then bound that value crudely using $k\geq2$.

    By the closed form (Equation~\ref{eq:2sgame:b1-det}), ${b^{k,\gamma}_{\dtr}/R=\bigl((k-1)^{k-1}/(\gamma k^k)\bigr)^{1/k}}$ is proportional to~$\gamma^{-1/k}$, hence strictly decreasing in $\gamma$.
    Raising $\gamma$ past $1$ therefore only lowers the bid:
    \begin{equation}\label{eq:game:det-bid-decreasing}
      b^{k,\gamma}_{\dtr}
      <
      b^{k,1}_{\dtr}.
    \end{equation}

    We next bound $b^{k,1}_{\dtr}$.
    Because $k\geq2$, we have $k-1\geq1$, so $(k-1)^{k-1}\leq(k-1)^k$.
    Therefore,
    \begin{equation}\label{eq:game:det-bid-upper}
      \frac{b^{k,1}_{\dtr}}{R}
      =
      \left(\frac{(k-1)^{k-1}}{k^k}\right)^{1/k}
      \leq
      \frac{k-1}{k}.
    \end{equation}

    Combining Equation~\ref{eq:game:det-bid-decreasing} and Equation~\ref{eq:game:det-bid-upper}, we have
    $$b^{k,\gamma}_{\dtr}<\frac{k-1}{k}R.$$
    Therefore, we have a lower bound on the deterrence utility:
    \[
      u^{k,\gamma}_{\dtr}
      =
      R-b^{k,\gamma}_{\dtr}
      >
      R-\frac{k-1}{k}R
      =
      \frac{R}{k}.
    \]
    Finally, since $\gamma>1$, we have $\min\{1,\gamma\}=1$, so that ${u^{k,\gamma}_{\dtr}>R/(k\,\min\{1,\gamma\})}$.

    Now suppose $0<\gamma\leq1$.
    Our goal is to show that $u^{k,\gamma}_{\dtr}>R/(k\,\min\{1,\gamma\})$, which in this case is~${u^{k,\gamma}_{\dtr}>R/(k\gamma)}$, whenever $k\geq\exp(1/\gamma)/\gamma$.

    Denoting $z\coloneqq k\gamma$, the hypothesis becomes ${z\geq\exp(1/\gamma)>1}$.
    The idea is to reduce this utility bound to a one-variable inequality in~$z$, monotone in~$z$, and then verify it at the smallest admissible value $z=\exp(1/\gamma)$.
    By the deterrence bid (Equation~\ref{eq:2sgame:b1-det}),
    \begin{align}\label{eq:game:det-bid-z}
      \frac{b^{k,\gamma}_{\dtr}}{R}
      &\;=\;
      \left(\frac{(k-1)^{k-1}}{\gamma k^k}\right)^{1/k} \nonumber\\
      &\;=\;
      \exp\!\left(\frac{(k-1)\ln(k-1)-\ln\gamma-k\ln k}{k}\right) \nonumber\\
      &\;=\;
      \exp\!\left(\frac{(k-1)\ln(1-1/k)-\ln k-\ln\gamma}{k}\right) \nonumber\\
      &\;=\;
      \exp\!\left(\frac{(k-1)\ln(1-1/k)-\ln z}{k}\right).
    \end{align}
    Since $\ln(1-1/k)\leq -1/k$, we have $$(k-1)\ln(1-1/k)\leq -(1-1/k).$$
    Combining this with Equation~\ref{eq:game:det-bid-z}, we have
    \[
      \frac{b^{k,\gamma}_{\dtr}}{R}
      \;\leq\;
      \exp\!\left(\frac{-(1-1/k)-\ln z}{k}\right).
    \]
    Define
    \[
      d\coloneqq \frac{\ln z+1-1/k}{k}.
    \]
    Since $-(1-1/k)-\ln z = -(\ln z+1-1/k)$, the exponent equals $-d$, so
    \[
      \frac{b^{k,\gamma}_{\dtr}}{R}
      \;\leq\;
      \exp(-d).
    \]
    It remains to turn this bound on the bid into the utility bound in Equation~\ref{eq:game:sufficient}.
    Since $u^{k,\gamma}_{\dtr}/R=1-b^{k,\gamma}_{\dtr}/R$, the bound $b^{k,\gamma}_{\dtr}/R\leq\exp(-d)$ gives
    \[
      \frac{u^{k,\gamma}_{\dtr}}{R} \;\geq\; 1-\exp(-d).
    \]
    Our next goal is therefore to prove
    \begin{equation}\label{eq:game:det-revenue-z}
      1-\exp(-d) \;>\; \frac{1}{z};
    \end{equation}
    together with $u^{k,\gamma}_{\dtr}/R\geq 1-\exp(-d)$, this gives the desired bound $u^{k,\gamma}_{\dtr}/R>1/z$.

    We weaken the left-hand side of Equation~\ref{eq:game:det-revenue-z} using the elementary inequality $1-\exp(-d)\geq d/(1+d)$, which holds for every $d\geq0$.
    So it suffices to prove
    \begin{equation*}
      \frac{d}{1+d} \;>\; \frac{1}{z}.
    \end{equation*}
    Multiplying through by the positive quantities $z$ and $1+d$, this is equivalent to $(z-1)d>1$.

    It is convenient to express $(z-1)d$ in closed form.
    Recall that $d=(\ln z+1-1/k)/k$, and that $k=z/\gamma$, so that $1/k=\gamma/z$.
    Substituting both occurrences of $1/k$ gives
    \[
      d \;=\; \frac{\gamma}{z}\left(\ln z+1-\frac{\gamma}{z}\right).
    \]
    Multiplying by $z-1$ and using $(z-1)/z=1-1/z$, we obtain
    \begin{equation*}
      F_\gamma(z)
      \;\coloneqq\; (z-1)\,d
      \;=\; \gamma\left(\ln z+1-\frac{\gamma}{z}\right)\left(1-\frac{1}{z}\right).
    \end{equation*}
    Our goal is now to show that $F_\gamma(z)>1$ for every~${z\geq\exp(1/\gamma)}$.

    We first show that $F_\gamma$ is strictly increasing in $z>1$.
    Differentiating the closed form gives
    \begin{equation*}
      \frac{F_\gamma'(z)}{\gamma}
      \;=\; \left(\frac{1}{z}+\frac{\gamma}{z^2}\right)\left(1-\frac{1}{z}\right)
      + \frac{\ln z+1-\gamma/z}{z^2} .
    \end{equation*}
    The first product is positive, since $z>1$ makes both of its factors positive.
    In the second term, the numerator~$\ln z+1-\gamma/z$ is increasing in $z$, and at $z=1$ it equals~$1-\gamma$, which is non-negative because $\gamma\leq1$.
    Hence the numerator is non-negative for all $z\geq1$.
    Both terms are therefore non-negative and the first is strictly positive, so~$F_\gamma'(z)>0$.

    Because $F_\gamma$ is increasing, over the range $z\geq\exp(1/\gamma)$ it is smallest at the left endpoint $z=\exp(1/\gamma)$.
    It is therefore enough to check that $F_\gamma(\exp(1/\gamma))>1$.
    Set $t\coloneqq1/\gamma$, so that $t\geq1$ because $\gamma\leq1$.
    Substituting $z=\exp(t)$ and~$\gamma=1/t$ into the closed form and simplifying yields
    \begin{equation*}
      F_\gamma(\exp(1/\gamma))-1
      \;=\; \gamma^2\exp(-t)\,g(t),
    \end{equation*}
    where
    \begin{equation*}
      g(t)\;\coloneqq\; t\exp(t)+\exp(-t)-t^2-t-1 .
    \end{equation*}
    Since $\gamma^2\exp(-t)>0$, it remains to show that $g(t)>0$ for~$t\geq1$, which we build up from its derivatives.
    The second derivative
    \[
      g''(t)\;=\;(2+t)\exp(t)+\exp(-t)-2
    \]
    is positive for $t\geq1$, so $g'$ is increasing on $[1,\infty)$.
    At the left endpoint, $g'(1)=2e-1/e-3>0$, so $g'$ is positive throughout $[1,\infty)$, and hence $g$ is increasing there.
    Finally $g(1)=e+1/e-3>0$, so $g(t)\geq g(1)>0$ for all~$t\geq1$.
    This gives $F_\gamma(\exp(1/\gamma))>1$, and by monotonicity $F_\gamma(z)>1$ for every $z\geq\exp(1/\gamma)$.

    Collecting the inequalities, for every $z\geq\exp(1/\gamma)$ we have
    \begin{equation*}
      \begin{aligned}
        \frac{u^{k,\gamma}_{\dtr}}{R}
        &\;=\; 1-\frac{b^{k,\gamma}_{\dtr}}{R}
        \;\geq\; 1-\exp(-d) \\
        &\;\geq\; \frac{d}{1+d}
        \;>\; \frac{1}{z}
        \;=\; \frac{1}{k\gamma}
        \;=\; \frac{1}{k\,\min\{1,\gamma\}},
      \end{aligned}
    \end{equation*}
    where the strict inequality is precisely $(z-1)d=F_\gamma(z)>1$.

    In all, when $k\geq\max\{2,\exp(1/\gamma)/\gamma\}$, we have ${u^{k,\gamma}_{\dtr}>R/(k\,\min\{1,\gamma\})}$ for every $\gamma>0$.
  \end{proof}

\section{Proof of Lemma~\ref{lem:game:positive-state-existence}}\label{app:proof:positive-state-existence}

\begin{restated}{Lemma}{lem:game:positive-state-existence}{SPNE existence after a positive user bid}
Consider a subgame starting from step $t$ with any history $h^{t-1}$ such that
\[
x_1(h^{t-1})=b_1^{(\lfloor t/2\rfloor)}>0.
\]
Then the subgame admits a pure-strategy SPNE.
\end{restated}

\begin{proof}
% \begin{fullproof}
Consider the subgame that starts from a history $h^{t-1}$ at the beginning of step $t$ at which $P_1$'s latest bid is positive, $x_1(h^{t-1})>0$.
By Hellwig and Leininger~\cite[Theorem~1]{hellwig1987existence}, a game has a pure-strategy SPNE provided that (i) it has a finite horizon, (ii) it has perfect information, (iii) the action sets are compact, (iv) the utility functions are continuous, and (v) the constraint correspondence, which assigns to each history the set of actions then feasible, is closed-valued and continuous in the history.

The subgame has a finite horizon and perfect information, so conditions~(i) and~(ii) hold.
The acting player~$P_i$ chooses a bid from $\mathbb{R}_{\ge0}$, and since bids are non-decreasing her feasible set is the closed interval $[x_i(h),\infty)$, where $x_i(h)$ is her latest bid.
This interval is closed, and its lower endpoint $x_i(h)$ varies continuously with the history, so the constraint correspondence is closed-valued and continuous in the history, and condition~(v) holds.
The action set $\mathbb{R}_{\ge0}$, however, is not compact, so condition~(iii) fails and the theorem does not apply to the subgame directly.
We proceed in three steps.
Step~1 verifies condition~(iv), the continuity of the utilities.
Step~2 constructs a new game with compact action sets, satisfying all of the conditions~(i)--(v), and applies the theorem to establish the existence of a pure-strategy SPNE of that game.
Step~3 extends this equilibrium to the original subgame by induction.

\emph{Step 1: Continuity of the utilities.}
Because bids are non-decreasing, $P_1$'s final bid satisfies
\[
b_1^{(m)}\ge x_1(h^{t-1})>0,
\]
and consequently
\[
(b_1^{(m)})^k+(b_2^{(m)})^k
\ge
\bigl(x_1(h^{t-1})\bigr)^k
>0.
\]
The total weight $W$ is therefore strictly positive, so the winning probabilities are continuous in the final bids.
By Equation~\ref{eq:2sgame:utility}, the utility functions are continuous.

\emph{Step 2: Construction of a compact game.}
We now construct a new game by capping each bid, so that the unbounded action sets become compact intervals.
We choose each cap so high that exceeding it is never worthwhile, which keeps the equilibria unchanged.
How high this upper bound must be follows from comparing two options for the acting player: raising her bid, or keeping her latest bid unchanged.
We denote
\[
\underline{\gamma}\coloneqq \min\{\gamma,1\}
 \text{, and }
\overline{\gamma}\coloneqq \max\{\gamma,1\}.
\]
Suppose $P_i$ acts at some step, with history~$h$ and latest bid~$x_i(h)$.
Keeping her bid at $x_i(h)$ through all of her remaining steps guarantees her a utility of at least $-\overline{\gamma}x_i(h)$: If her final bid remains $x_i(h)$, then by Equation~\ref{eq:2sgame:utility} her utility at any final winning probability $q_i\in[0,1]$ is
\[
\begin{aligned}
u_i
&=
q_i\bigl(R-x_i(h)\bigr)-(1-q_i)\gamma x_i(h) \\
&=
q_iR-\bigl[\gamma+(1-\gamma)q_i\bigr]x_i(h),
\end{aligned}
\]
which is at least $-\overline{\gamma}x_i(h)$ because $\gamma+(1-\gamma)q_i\le \overline{\gamma}$ and $q_iR\ge 0$.
Suppose instead that she raises her bid to some $b'$ at this step.
Let $\widehat b_i$ be her resulting final bid; since bids are non-decreasing, $\widehat b_i\ge b'$.
Equation~\ref{eq:2sgame:utility} gives
\[
u_i
=
q_iR-\bigl[\gamma+(1-\gamma)q_i\bigr]\widehat b_i
\le
R-\underline{\gamma}\widehat b_i
\le
R-\underline{\gamma}b',
\]
where the first inequality uses $\gamma+(1-\gamma)q_i\ge \underline{\gamma}$.
Hence, as soon as
\[
b'>\frac{R+\overline{\gamma}x_i(h)}{\underline{\gamma}},
\]
we have $u_i\le R-\underline{\gamma}b'<-\overline{\gamma}x_i(h)$, so raising the bid to~$b'$ is strictly worse than keeping $x_i(h)$.
We call this the \emph{exclusion bound}: Any bid above $(R+\overline{\gamma}x_i(h))/\underline{\gamma}$ can be safely excluded.

We turn this principle into a compact game by induction on the number of remaining steps
\[
r\coloneqq 2m-t+1.
\]
The base case $r=0$ leaves no step to play, so the claim vacuously holds.
For the inductive step, suppose $r\ge 1$ and that the claim holds for every subgame in which $P_1$'s latest bid is positive, with fewer than $r$ remaining steps.

Guided by the exclusion bound, we construct a new game, the \emph{restricted game}, that coincides with the subgame except that every bid is capped, so that the action sets become compact intervals.
The caps are a deterministic sequence of bounds
\[
B_0,B_1,\ldots,B_r,
\]
where $B_\ell$ bounds both players' latest bids after $\ell$ of the remaining steps have been played.
Starting from the latest bids, we set
\[
B_0\coloneqq \max\{x_1(h^{t-1}),x_2(h^{t-1})\}.
\]
For $\ell=0,\ldots,r-1$, define
\[
B_{\ell+1}\coloneqq
\frac{R+\overline{\gamma}B_\ell}{\underline{\gamma}},
\]
the exclusion bound's threshold with the latest bid $x_i(h)$ replaced by its upper bound $B_\ell$.
Because $R>0$, $\underline{\gamma}\le 1$, and $\overline{\gamma}\ge 1$, we have
\[
B_{\ell+1}>B_\ell
\]
for every $\ell$.

Consider the $(\ell+1)$-st remaining step, where $\ell=0,\ldots,r-1$.
Let the latest bid pair of this step be $(x_1,x_2)$.
In the restricted game, we maintain the invariant
\[
x_1(h^{t-1})\le x_1\le B_\ell,
\qquad
0\le x_2\le B_\ell.
\]
The invariant holds at $\ell=0$ by the definition of $B_0$.
If it holds before the $(\ell+1)$-st remaining step, then the acting player is allowed to choose a bid only up to $B_{\ell+1}$, while the non-acting player's bid is unchanged.
Hence the invariant holds at the next step as well.

More explicitly, the $(\ell+1)$-st remaining step is step number
\[
s\coloneqq t+\ell
\]
of the full $2m$-step game.
If $s$ is odd, then $P_1$ acts and her feasible bids in the restricted game are
\[
[x_1,B_{\ell+1}].
\]
If $s$ is even, then $P_2$ acts and her feasible bids in the restricted game are
\[
[x_2,B_{\ell+1}].
\]
The action set at each step is the compact interval $[0,B_{\ell+1}]$, and the constraint correspondence assigns the nonempty closed subinterval $[x_i,B_{\ell+1}]$.
This correspondence is continuous, because its lower endpoint is the acting player's latest bid and its upper endpoint is fixed at that step.
The state transition is also continuous: After a bid $b$ by $P_1$, the new latest bid pair is $(b,x_2)$; after a bid $b$ by $P_2$, it is $(x_1,b)$.

Together with the continuity of the utilities from Step~1, the restricted game is a finite-horizon game of perfect information with compact action sets, a closed-valued continuous constraint correspondence, and continuous utilities.
By Theorem~1 of Hellwig and Leininger~\cite{hellwig1987existence}, it admits a pure-strategy SPNE.

\emph{Step 3: Extension to the unbounded subgame.}
It remains to show that this restricted-game SPNE is also an SPNE of the original unbounded subgame.
Consider any history $h$ inside the compact region before the $(\ell+1)$-st remaining step, and let $P_i$ be the acting player with the latest bid $x_i(h)$, so that $x_i(h)\le B_\ell$ by the invariant.
Suppose she deviates outside the restricted feasible set to some bid $b>B_{\ell+1}$.
Then
\[
b>B_{\ell+1}
=
\frac{R+\overline{\gamma}B_\ell}{\underline{\gamma}}
\ge
\frac{R+\overline{\gamma}x_i(h)}{\underline{\gamma}},
\]
so by the exclusion bound of Step~2 this deviation gives her utility strictly below $-\overline{\gamma}x_i(h)$.
Keeping her bid at $x_i(h)$, on the other hand, is feasible in the restricted game and guarantees her at least $-\overline{\gamma}x_i(h)$.
Hence no excluded bid is a profitable deviation.

We now assemble an equilibrium of the original unbounded subgame.
On histories that remain inside the compact region, use the SPNE of the restricted game.
If a player nevertheless chooses a bid outside the compact region, the resulting subgame has fewer than $r$ remaining steps, and $P_1$'s latest bid is still positive: If $P_1$ deviates, her own bid remains positive; if $P_2$ deviates, $P_1$'s positive bid is unchanged.
By the induction hypothesis, fix a pure-strategy SPNE for each such off-path subgame.

The resulting strategy profile is sequentially optimal at every history.
At histories inside the compact region, deviations within the compact action set are unprofitable by the restricted-game SPNE, while deviations outside the compact action set are unprofitable by the exclusion bound.
At histories outside the compact region, sequential optimality follows from the induction hypothesis.
Hence the constructed strategy profile is an SPNE of the original subgame.

This completes the induction and proves the lemma.
% \end{fullproof}
\end{proof}

\section{Proof of Lemma~\ref{lem:game:positive-entry-bound}}\label{app:proof:positive-entry-bound}

\begin{restated}{Lemma}{lem:game:positive-entry-bound}{Positive-entry continuation bound}
Consider a subgame immediately after $P_1$ publishes a positive bid from an all-zero state, and let $\bigl(b_1^{(m)},b_2^{(m)}\bigr)$ be the final bids under any SPNE of this subgame.
Then
\[
u_1\bigl(b_1^{(m)},b_2^{(m)}\bigr)
\le
\max\left\{
u^{k,\gamma}_{\dtr},
u^{k,\gamma}_{\mathrm{acc}}
\right\}.
\]
\end{restated}

\begin{proof}
% \begin{fullproof}
We proceed in three steps.
Step~1 bounds $P_2$'s final bid below by her smallest best response over all non-negative bids.
Step~2 shows that lowering $P_2$'s final bid to that smallest best response does not decrease $P_1$'s utility, reducing the bound to a single two-step outcome.
Step~3 bounds this two-step outcome by either the deterrence or the accommodation utility.

Before carrying out these steps, we dispose of a trivial case.
If $P_1$'s utility is nonpositive, then the bound is immediate:
\[
u_1\bigl(b_1^{(m)},b_2^{(m)}\bigr)\le 0<u_{\dtr}^{k,\gamma}\le\max\bigl\{u_{\dtr}^{k,\gamma},u_{\mathrm{acc}}^{k,\gamma}\bigr\},
\]
where $u_{\dtr}^{k,\gamma}>0$ by Proposition~\ref{prop:det-profitable}.
It therefore remains to consider the case
\begin{equation}\label{eq:game:positive-entry-positive-utility}
u_1\bigl(b_1^{(m)},b_2^{(m)}\bigr)>0.
\end{equation}

\emph{Step 1: A lower bound on $P_2$'s final bid.}
Recall that $b_2^{(m-1)}$ is $P_2$'s bid immediately before her final action.
Since the strategy is subgame perfect, $P_2$'s final bid solves the constrained best-response problem
\[
b_2^{(m)}
\in
\operatorname*{arg\,max}_{b_2\ge b_2^{(m-1)}}
u_2\bigl(b_1^{(m)},b_2\bigr).
\]
We now consider $P_2$'s best response to $b_1^{(m)}$ over all non-negative bids, which coincides with her best response in the two-step game.
Because $P_1$ has published a positive bid and bids are non-decreasing, we have $b_1^{(m)}>0$.
The two-step analysis (\S\ref{sec:game:2seq:p2}) shows that $P_2$'s best-response correspondence $\operatorname{BR}_2\bigl(b_1^{(m)}\bigr)$ is nonempty and compact.
Denote its smallest element by
\[
\underline b_2
\coloneqq
\min \operatorname{BR}_2\bigl(b_1^{(m)}\bigr).
\]

We now show that $\underline b_2$ is a lower bound on $b_2^{(m)}$, that is, $\underline b_2\le b_2^{(m)}$.
If $\underline b_2$ lies below $P_2$'s minimal feasible final bid, that is, $\underline b_2<b_2^{(m-1)}$, the claim is immediate from $b_2^{(m)}\ge b_2^{(m-1)}$.
Otherwise $\underline b_2\ge b_2^{(m-1)}$, so $\underline b_2$ is feasible in the constrained final-step problem.
Since $\underline b_2$ already maximizes $u_2(b_1^{(m)},\cdot)$ over all non-negative bids, the constrained and unconstrained maxima coincide, so $b_2^{(m)}$ is itself an unconstrained best response:
\[
b_2^{(m)}
\in
\operatorname{BR}_2\bigl(b_1^{(m)}\bigr).
\]
By the minimality of $\underline b_2$, we again obtain
\[
\underline b_2\le b_2^{(m)}.
\]

\emph{Step 2: Lowering $P_2$'s final bid does not decrease $P_1$'s utility.}
We show that replacing $b_2^{(m)}$ by $\underline b_2$ does not decrease $P_1$'s utility $u_1$.
Using the utility formula (Equation~\ref{eq:2sgame:utility}), we obtain
\begin{multline*}
u_1\bigl(b_1^{(m)},b_2^{(m)}\bigr)
= \\
 -\gamma b_1^{(m)}  + q_1\bigl(b_1^{(m)},b_2^{(m)}\bigr)\left[R-(1-\gamma)b_1^{(m)}\right].
\end{multline*}
By Equation~\ref{eq:game:positive-entry-positive-utility}, $u_1\bigl(b_1^{(m)},b_2^{(m)}\bigr)>0$, which is possible only if
\begin{equation}\label{eq:game:positive-entry-coeff}
R-(1-\gamma)b_1^{(m)}>0.
\end{equation}
Since $\underline b_2\le b_2^{(m)}$, lowering $P_2$'s final bid weakly increases~$P_1$'s winning probability:
\[
q_1\bigl(b_1^{(m)},\underline b_2\bigr)
\ge
q_1\bigl(b_1^{(m)},b_2^{(m)}\bigr).
\]
Together with Equation~\ref{eq:game:positive-entry-coeff}, this implies
\[
u_1\bigl(b_1^{(m)},b_2^{(m)}\bigr)
\le
u_1\bigl(b_1^{(m)},\underline b_2\bigr).
\]

\emph{Step 3: Comparison with the deterrence and accommodation utilities.}
It remains to show that
\[
u_1\bigl(b_1^{(m)},\underline b_2\bigr)\le\max\bigl\{u_{\dtr}^{k,\gamma},u_{\mathrm{acc}}^{k,\gamma}\bigr\}.
\]
We treat separately the two cases $b_1^{(m)}\ge b_{\dtr}^{k,\gamma}$ and $0<b_1^{(m)}<b_{\dtr}^{k,\gamma}$, which yield the bounds $u_{\dtr}^{k,\gamma}$ and $u_{\mathrm{acc}}^{k,\gamma}$, respectively.

If $b_1^{(m)}\ge b_{\dtr}^{k,\gamma}$, using $q_1\bigl(b_1^{(m)},\underline b_2\bigr)\le 1$ and Equation~\ref{eq:game:positive-entry-coeff}, we obtain
\[
\begin{aligned}
u_1\bigl(b_1^{(m)},\underline b_2\bigr)
&=
-\gamma b_1^{(m)}
+
q_1\bigl(b_1^{(m)},\underline b_2\bigr)
\left[
R-(1-\gamma)b_1^{(m)}
\right] \\
&\le
-\gamma b_1^{(m)}
+
R-(1-\gamma)b_1^{(m)} \\
&=
R-b_1^{(m)} \\
&\le
R-b_{\dtr}^{k,\gamma}\\
&=
u_{\dtr}^{k,\gamma}.
\end{aligned}
\]
If instead $0<b_1^{(m)}<b_{\dtr}^{k,\gamma}$, then $\underline b_2 = \min\operatorname{BR}_2\bigl(b_1^{(m)}\bigr)$ is by definition $P_2$'s selected best response $b_2^*\bigl(b_1^{(m)}\bigr)$~(\S\ref{sec:game:2seq:p2}).
Since $b_1^{(m)}$ lies in the accommodation range, the definition of $u_{\mathrm{acc}}^{k,\gamma}$ as a supremum over that range (Equation~\ref{eq:2sgame:u1-acc}) yields
\[
u_1\bigl(b_1^{(m)},\underline b_2\bigr)=u_1\bigl(b_1^{(m)},b_2^*(b_1^{(m)})\bigr)\le u_{\mathrm{acc}}^{k,\gamma}.
\]
Combining the two cases, we have
\[
u_1\bigl(b_1^{(m)},\underline b_2\bigr)\le\max\bigl\{u_{\dtr}^{k,\gamma},u_{\mathrm{acc}}^{k,\gamma}\bigr\}.
\]

\emph{Summary.}
Chaining the conclusions of Steps~2 and~3, we obtain
\[
u_1\bigl(b_1^{(m)},b_2^{(m)}\bigr)\le u_1\bigl(b_1^{(m)},\underline b_2\bigr)\le\max\bigl\{u_{\dtr}^{k,\gamma},u_{\mathrm{acc}}^{k,\gamma}\bigr\},
\]
which proves the lemma.
% \end{fullproof}
\end{proof}

\section{Proof of Theorem~\ref{thm:game:multistep}}\label{app:proof:multistep}

\begin{restated}{Theorem}{thm:game:multistep}{Deterrence from every all-zero subgame}
Suppose $u_{\dtr}^{k,\gamma}>u_{\mathrm{acc}}^{k,\gamma}$.
For every odd step $t=2j-1$, consider an all-zero subgame beginning at step $t$, that is,
\[
b_1^{(\lfloor t/2\rfloor)}
=
b_1^{(j-1)}
=
0
\]
and
\[
b_2^{(\lfloor (t-1)/2\rfloor)}
=
b_2^{(j-1)}
=
0.
\]
Then this all-zero subgame admits an SPNE in which
\[
b_1^{(\ell)}=b_{\dtr}^{k,\gamma}
\quad\text{and}\quad
b_2^{(\ell)}=0
\qquad
\text{for every }\ell=j,\ldots,m.
\]
\end{restated}

\begin{proof}
% \begin{fullproof}
We proceed by backward induction on $j$.
When~${j=m}$, the all-zero subgame consists only of steps $2m-1$ and $2m$.
The claim is therefore exactly the two-step deterrence equilibrium in Theorem~\ref{thm:2sgame:eq}.

For the induction step, fix $j<m$ and assume the theorem holds for index $j+1$, that is, for the all-zero subgame beginning at step $2(j+1)-1=2j+1$.

Before the concrete induction analysis, we first specify a continuation, the actions the players take in the subgame afterwards, to be used after a deterring bid has been published.
Consider any history $h$ at which $P_1$'s latest bid satisfies~$x_1(h)\ge b_{\dtr}^{k,\gamma} $ and $P_2$'s latest bid satisfies~$x_2(h)=0$.
Prescribe that both players leave their bids unchanged for the rest of the game.
If $P_2$ deviates to any positive bid, then switch to a pure-strategy SPNE of the resulting positive-user subgame, whose existence is guaranteed by Lemma~\ref{lem:game:positive-state-existence}.

We verify that this prescription is sequentially optimal.
If $P_2$ deviates to a positive bid, then under any continuation her final bid remains positive, and $P_1$'s final bid remains at least $b_{\dtr}^{k,\gamma}$ because bids are non-decreasing.
Therefore the final weight of $P_1$ is at least
$\bigl(b_{\dtr}^{k,\gamma}\bigr)^k
=
W_{\dtr}^{k,\gamma}.$

By Lemma~\ref{lem:2sgame:deterrence}, $P_2$'s terminal utility is at most zero, regardless of her final positive bid.
Staying at zero also gives her utility zero.
Hence $P_2$ has no profitable deviation.
While~$P_2$ stays at zero, $P_1$ wins for sure and earns $R-x_1(h)$.

Since bids cannot be decreased, raising $P_1$'s bid from $x_1(h)$ only lowers this utility.
Thus $P_1$ also has no profitable deviation.
Therefore this prescription is an SPNE of every such deterring subgame.

We now carry out the concrete analysis for the all-zero subgame beginning at step $2j-1$.
We prescribe that~$P_1$ chooses~$b_1^{(j)}=b_{\dtr}^{k,\gamma}$ and then follows the deterring continuation described above.
This gives $P_1$ utility
$u_{\dtr}^{k,\gamma} = R-b_{\dtr}^{k,\gamma}$.

We now check that $P_1$ has no profitable deviation at step~$2j-1$.
There are three relevant cases.

First, if she follows the prescribed strategy and publishes~$b_{\dtr}^{k,\gamma}$, she obtains $u_{\dtr}^{k,\gamma}$ by the deterring continuation.

Second, if she chooses $b_1^{(j)}=0$, then $P_2$ has not observed any transaction and therefore has no feasible positive entry action at step~$2j$, i.e., $b_2^{(j)}=0$.
The game then reaches the all-zero subgame beginning at step~$2j+1$, and the induction hypothesis gives $P_1$ utility~$u_{\dtr}^{k,\gamma}$.

Third, if she chooses some other positive bid $b_1^{(j)}>0$, then the game enters a positive-user subgame.
For that continuation, choose a pure-strategy SPNE supplied by Lemma~\ref{lem:game:positive-state-existence}.
Let
\[
\bigl(b_1^{(m)},b_2^{(m)}\bigr)
\]
be the final bids under this continuation.
By Lemma~\ref{lem:game:positive-entry-bound},
\[
u_1\bigl(b_1^{(m)},b_2^{(m)}\bigr)
\le
\max\left\{
u_{\dtr}^{k,\gamma},
u_{\mathrm{acc}}^{k,\gamma}
\right\}
=
u_{\dtr}^{k,\gamma},
\]
where the equality uses the hypothesis
\[
u_{\dtr}^{k,\gamma}>u_{\mathrm{acc}}^{k,\gamma}.
\]
Thus no positive deviation can give $P_1$ more than the prescribed deterring bid.

It remains only to specify the continuation after histories not covered by the prescribed path.
If the play reaches a later all-zero subgame, we use the SPNE given by the induction hypothesis.
If the play reaches a positive-user subgame, we use an SPNE supplied by Lemma~\ref{lem:game:positive-state-existence}, except in the deterring continuation described above, which has already been shown to be an SPNE.
Therefore every continuation following the prescribed action is subgame perfect, and $P_1$ has no profitable deviation at the current all-zero state.
The prescribed strategies consequently form an SPNE of the all-zero subgame beginning at step $2j-1$.

This completes the backward induction.
% \end{fullproof}
\end{proof}

\fi

  \section{Numerical Simulation for $k$ Values}\label{app:empirical-k}

  We numerically find the exponent parameter $k$ that makes PRECEDE Anti-Spam and deterrence the equilibrium.
  For a given $\gamma$, a valid $k$ satisfies the deterrence condition $u^{k,\gamma}_{\dtr} > u^{k,\gamma}_{\mathrm{acc}}$ (Theorem~\ref{thm:2sgame:eq}), the Anti-Spam bound $k\ge\ln 2/\ln(1+\gamma)$ (Theorem~\ref{thm:sybil:tight}), and $k>1$ by PRECEDE's definition~(\S\ref{sec:precede:ordering}).
  Since we have no closed form for $u^{k,\gamma}_{\mathrm{acc}}$, we evaluate these conditions numerically.

  We target a range $\gamma\in[\gamma_{\min},\gamma_{\max}]$ and seek $k$ values that work across it.
  We set this range from our empirical measurements of $\gamma$ (\S\ref{app:practical:gamma}), in two cases: to secure against front-running on Uniswap V3 alone, $[0.867,1.087]$, and to secure both V2 and V3, $[0.259,1.087]$.

  Because the deterrence advantage $u^{k,\gamma}_{\dtr}-u^{k,\gamma}_{\mathrm{acc}}$ need not be monotone in $k$, checking the smallest feasible $k$ alone does not suffice: A larger $k$ may fail where a smaller one works.
  We therefore require that \emph{every} $k$ above $k_{\min}$ meets all three conditions over the \emph{entire} $\gamma$ range, and we find the smallest such $k_{\min}$.
  A numerical search must be bounded, so we check the conditions across $[k_{\min},k_{\max}]$ with $k_{\max}=20$, and make no claim for $k>k_{\max}$.
  We set~$R=1$ without loss of generality, since all revenues scale with it.

  We find $k_{\min}$ by bisection, starting from the Anti-Spam bound $\ln 2/\ln(1+\gamma_{\min})$, the smallest $k$ the range admits, verifying for each candidate that all conditions hold across the rectangle ${[\gamma_{\min},\gamma_{\max}]\times[k_{\min},k_{\max}]}$, where we search for the worst point by differential evolution, a continuous global optimizer.
  We then verify the resulting rectangle twice over: On a uniform~$101\times101$ grid, and by five differential-evolution searches, which explore the rectangle continuously to catch violations hiding between grid points.

  For $\gamma \in [0.867,1.087]$, the search yields $k_{\min}\approx 1.146$, above the Anti-Spam bound $\approx 1.111$, so the deterrence condition is what pushes $k_{\min}$ up.
  For ${\gamma \in [0.259,1.087]}$, it yields $k_{\min}\approx 3.010$, the Anti-Spam bound itself, so the search stops where it starts and Anti-Spam alone determines~$k_{\min}$.
  Both $k_{\min}$ values are rounded up, since $k_{\min}$ is a lower bound on the admissible $k$.
  \begin{showWhenSubmit}
    The code is in our anonymous repo~\cite{artifact}.
  \end{showWhenSubmit}

\iffullversion

\section{Proof of Proposition~\ref{prop:tradeoff:optimal-k}}\label{app:proof:optimal-k}

\begin{restated}{Proposition}{prop:tradeoff:optimal-k}{Revenue-optimal $k$}
    Fix $\gamma > 0$. The user's deterrence revenue $u^{k,\gamma}_{\dtr}(k,\gamma)$, viewed as a function of ${k>1}$, is strictly increasing on $\bigl(1,\,1+\frac{1}{\gamma}\bigr)$ and strictly decreasing on $\bigl(1+\frac{1}{\gamma},\infty\bigr)$. Its maximum, attained uniquely at $k = 1 + \frac{1}{\gamma}$, is
    \begin{equation*}
      u^{k,\gamma}_{\dtr}\Bigl(1+\tfrac{1}{\gamma},\,\gamma\Bigr) \;=\; \frac{\gamma}{1+\gamma}\,R .
    \end{equation*}
\end{restated}

  \begin{proof}
% \begin{fullproof}
    Differentiating $u^{k,\gamma}_{\dtr}$ in $k$ directly is cumbersome, so we instead study the bid $b^{k,\gamma}_{\dtr}$ and recover the revenue through $u^{k,\gamma}_{\dtr} = R - b^{k,\gamma}_{\dtr}$. 
    Let
    \begin{equation*}
      \begin{aligned}
        L(k) &\coloneqq \ln\frac{b^{k,\gamma}_{\dtr}(k,\gamma)}{R} \\
             &= \frac{(k-1)\ln(k-1) - k\ln k - \ln\gamma}{k}.
      \end{aligned}
    \end{equation*}
    Differentiating term by term, the numerator satisfies
    \begin{equation*}
      \frac{d}{dk}\bigl[(k-1)\ln(k-1)-k\ln k\bigr] = \ln\!\Bigl(1-\tfrac{1}{k}\Bigr),
    \end{equation*}
    so that, after collecting terms,
    \begin{equation*}
      L'(k) \;=\; \frac{\ln\!\bigl(\gamma(k-1)\bigr)}{k^{2}} ,
    \end{equation*}
    so $L'(k) < 0$ for $k < 1+\frac{1}{\gamma}$ and $L'(k) > 0$ for ${k > 1+\frac{1}{\gamma}}$. Hence $b^{k,\gamma}_{\dtr}$ is minimized uniquely at $k = 1+\frac{1}{\gamma}$, and ${u^{k,\gamma}_{\dtr} = R - b^{k,\gamma}_{\dtr}}$ is correspondingly maximized there. At $k = 1+\frac{1}{\gamma}$, $k-1 = \frac{1}{\gamma}$ and $k = \frac{1+\gamma}{\gamma}$, so
    \begin{equation*}
      \begin{aligned}
        &(k-1)\ln(k-1) - k\ln k - \ln\gamma \\
        =& -\tfrac{1}{\gamma}\ln\gamma - \tfrac{1+\gamma}{\gamma}\bigl[\ln(1+\gamma) - \ln\gamma\bigr] - \ln\gamma \\
        =& -\tfrac{1+\gamma}{\gamma}\ln(1+\gamma) ,
      \end{aligned}
    \end{equation*}
    where the $\ln\gamma$ contributions cancel (their coefficients sum to $-\tfrac{1}{\gamma}+\tfrac{1+\gamma}{\gamma}-1 = 0$). Dividing by $k = \frac{1+\gamma}{\gamma}$ gives ${L = -\ln(1+\gamma)}$, i.e.\ $b^{k,\gamma}_{\dtr}\bigl(1+\tfrac{1}{\gamma},\,\gamma\bigr) = \frac{R}{1+\gamma}$.
    Therefore, we obtain the desired result:
    \[
      u^{k,\gamma}_{\dtr}\bigl(1+\tfrac{1}{\gamma},\,\gamma\bigr) = R - b^{k,\gamma}_{\dtr}  = \frac{\gamma R}{1+\gamma}. \qedhere
    \]
% \end{fullproof}
  \end{proof}

\fi

\end{document}